\documentclass[sigconf]{acmart}

\usepackage{times}
\usepackage{soul}
\usepackage{url}
\usepackage{graphicx}
\usepackage{amsmath}
\usepackage{booktabs}
\usepackage{amssymb}
\usepackage{amsthm}
\usepackage{mathrsfs}
\usepackage{subfig}
\usepackage{adjustbox} 
\usepackage[utf8]{inputenc}
\usepackage{multirow}
\usepackage{algorithm}
\usepackage{algpseudocode}
\usepackage{xspace}
\usepackage{amsfonts}
\usepackage{float}
\usepackage{enumitem}
\usepackage[most]{tcolorbox}
\usepackage{lipsum}
\usepackage{float}
\usepackage{makecell}
\newtheorem{theorem}{\bf Theorem}
\newcommand{\attack}{\textit{ActInv}\xspace}
\newcommand{\defense}{\textit{PriPert}\xspace}
\usepackage[dvipsnames, svgnames, x11names]{xcolor}

\newcommand{\result}[2]{%
  #1\textnormal{\smash{$_{\pm #2}$}}%
}

\usepackage{tikz}
\usepackage{amsmath}
\usepackage{xcolor}
\newcommand{\revise}[1]{\textcolor{black}{#1}}

\AtBeginDocument{%
  }

\setcopyright{acmlicensed}
\copyrightyear{2018}
\acmYear{2018}
\acmDOI{XXXXXXX.XXXXXXX}
\acmConference[Conference acronym 'XX]{Make sure to enter the correct
  conference title from your rights confirmation email}{June 03--05, 2018}{Woodstock, NY}
\acmISBN{978-1-4503-XXXX-X/2018/06}

\begin{document}

\title{What Does the Server See? Understanding Privacy Leakage from Large Language Models in Split Inference}

\author{Mingyuan Fan}
\affiliation{%
  \institution{East China Normal University}
  \city{Shanghai}
  \country{China}}
\email{fmy2660966@gmail.com}

\author{Yu Liu}
\affiliation{%
  \institution{East China Normal University}
  \city{Shanghai}
  \country{China}}
\email{yuliu@stu.ecnu.edu.cn}

\author{Fuyi Wang}
\affiliation{%
  \institution{RMIT University}
  \city{Melbourne}
  \country{Australia}}
\email{fuyi.wang@rmit.edu.au}

\author{Cen	Chen}
\authornote{Corresponding author.}
\affiliation{%
  \institution{East China Normal University}
  \city{Shanghai}
  \country{China}}
\email{cenchen@dase.ecnu.edu.cn}

\begin{abstract}
    The deployment of large language models (LLMs) on resource-constrained devices remains challenging, spurring interest in split inference, where models are partitioned between client and server to reduce computational burden and enhance privacy by transmitting only intermediate activations.
    However, the privacy-preserving capabilities of split inference, particularly in the context of LLMs, have not been exhaustively investigated.
    To fill this gap, we introduce \attack, which solves an intermediate activation matching problem to reconstruct the client's input.
    Extensive evaluations demonstrate that \attack achieves high-fidelity reconstructions, even in the presence of common perturbation-based defenses such as Gaussian noise injection and activation sparsification.
    
    To systematically understand this vulnerability, we develop Perturbation Amplification Factor (PAF), a metric for quantifying a layer's inherent resistance to reconstruction.
    Our analysis reveals that privacy vulnerability is not uniform across layers, with some layers being highly susceptible to leakage while others offer natural resistance.
    Furthermore, we demonstrate that defense effectiveness can be significantly improved by calibrating perturbation directions to maximize reconstruction error during backpropagation.
    Building on these insights, we design \defense and conduct comprehensive evaluations, covering privacy, utility, and computational overhead, to demonstrate its effectiveness.
\end{abstract}



\received{20 February 2007}
\received[revised]{12 March 2009}
\received[accepted]{5 June 2009}

\maketitle

\section{Introduction}
\label{sec1}
The rapid proliferation of large language models (LLMs) has revolutionized numerous applications~\cite{qwen3,gpt4,llama3}, yet their deployment on resource-constrained devices, such as smartphones and embedded systems, remains a significant challenge~\cite{SmolLM2,split_2,Teerapittayanon17}.
Modern LLMs contain billions of parameters and demand substantial computational and memory resources, rendering full on-device inference largely impractical~\cite{Teerapittayanon17}.
While recent advancements in model compression~\cite{QAT,LLM_Pruner,CoT_distll} have enabled the construction of smaller LLMs, these approaches often incur undesirable \revise{accuracy–efficiency} trade-offs.
At the same time, the proprietary nature of modern LLMs, coupled with the competitive advantage they confer, disincentivizes vendors from releasing fully deployable models~\cite{gpt4}, further limiting client-side deployment.

A natural alternative is to offload the entire inference process to the cloud~\cite{clould1,clould2}, where the client sends raw inputs to a server hosting the LLM and receives the outputs in return.
However, this indeed introduces severe privacy risks, particularly in privacy-sensitive domains like healthcare~\cite{split_2,abs-2409-19134,abs-1912-12115}.
For instance, if a user's medical queries or personal diary entries are sent directly to a cloud server for LLM processing, it would expose highly confidential information, raising serious privacy concerns.

Split inference has emerged as a promising middle ground~\cite{split_1,split_2,usenix_ref_1,usenix_ref_2,missing_1}.
This paradigm partitions an LLM into two non-overlapping components: a small initial segment residing on the client device, and a large, computationally intensive segment hosted on a remote server.
The client generates intermediate activations that are subsequently transmitted to the server.
The server then completes the remainder of the inference pipeline, with the final inference result sent back to the client.
Split inference offers three key advantages:
(1) It significantly alleviates the computational and storage burden on the client. 
(2) It enhances client privacy by transmitting only intermediate activations, rather than raw inputs, to the server.
(3) It allows model owners to protect their intellectual property, as the bulk of the model parameters and architecture remain on the server.

Despite these benefits, to the best of our knowledge, there remains a blank in the comprehensive investigation of the privacy landscape of LLMs within the split inference.
The assumption that intermediate activations inherently preserve privacy requires rigorous scrutiny, as these activations may still encode information about the original sensitive inputs.
In this paper, we study the privacy implications of LLMs under split inference by answering three key research questions (RQs).

\textit{\textbf{RQ 1: How can the privacy leakage of LLMs in split inference be efficiently quantified and evaluated?}}
We propose \attack, an attack method that reconstructs the client’s raw input from intermediate activations transmitted to the server.
\attack casts reconstruction as an optimization problem over the input embedding space: it searches for embeddings whose forward pass through the client-side submodel reproduces the intercepted activations at the cut layer, and then projects the optimized embeddings back to discrete tokens.
\revise{Unlike common attack methods that sequentially recover tokens via complex enumeration or gradient search, \attack optimizes the full sequence simultaneously, achieving comparable reconstruction accuracy with considerably less computation.}
We evaluate \attack across diverse datasets and LLMs in which \attack can reconstruct high-fidelity reconstructions, even when the intermediate activations are significantly perturbed.

\textit{\textbf{RQ 2: What factors govern the severity of privacy leakage?}}
To disentangle the root causes of privacy leakage, we introduce Perturbation Amplification Factor (PAF) to measures the invertibility hardness of a layer.
A high PAF indicates that small activation perturbations translate into large uncertainty in the inferred input, making accurate reconstruction difficult.
Using PAF, we uncover strong non-uniformity in privacy risk across layers: some layers naturally exhibit high PAF and thus resist reconstruction, whereas others present low PAF and are highly susceptible to leakage.
In particular, we find that many commonly used activation layers introduce low-PAF regions, making them easier to invert despite their benefits for model expressivity.
Moreover, we show that carefully aligned noise in activation space can induce much larger reconstruction errors than isotropic noise of the same magnitude, implying that naive perturbation defenses are suboptimal.

\textit{\textbf{RQ 3: How can we effectively mitigate privacy leakage with a small utility loss?}}
Building on these insights, we propose \defense, which injects adversarially calibrated perturbations into intermediate activations to obscure sensitive information while preserving model utility.
\defense formulates an optimization problem over perturbations with two competing objectives: maximizing reconstruction error for any plausible inversion (privacy) and minimizing the degradation of the server-side output (utility).
Through extensive experiments across multiple LLMs and datasets, we show that \defense achieves considerable improvements over baseline defenses in reducing \attack's reconstruction quality, with small degradation in downstream task performance.

\section{Related Work}
\label{sec6}

Since the release of the GPT series by OpenAI~\cite{gpt4}, LLMs have dominated the NLP landscape.
The architectures of modern LLMs have largely converged: they are predominantly based on autoregressive prediction paradigms and exhibit highly similar layer structures (Figure \ref{sec4_fig1})~\cite{llm_survey}.
The primary distinctions among them therefore arise from differences in training data and optimization strategies.

Research on privacy risks in LLMs can be divided into two stages: training and inference.
At the inference stage, most studies focus on risks of training data exposure, i.e., constructing adversarial prompts to elicit memorized training samples from LLMs.
Early research~\cite{paper9,paper13,paper25} demonstrated that LLMs tend to memorize portions of their pre-training corpus, raising the possibility of extracting training data.
Follow-up studies refined these findings by targeting specific categories of private data~\cite{paper21,paper22}, analyzing how training and decoding choices affect leakage~\cite{paper20,paper10}, and evaluating the risks posed by commercial deployments~\cite{paper14,paper24}.
Additional efforts highlighted broader misuse scenarios, such as extracting personal information directly from web data using LLMs~\cite{paper26}.

Orthogonal to this line of work, we investigate the privacy risks faced by end users of LLM inference services, particularly in split inference.
Split inference itself is a broad paradigm, motivated not only by privacy but also by practical concerns such as reducing communication overhead and enabling on-device computation~\cite{usenix_ref_1,usenix_ref_2}.
Moreover, split inference is closely related to split learning~\cite{paper7,missing_1}, where prior studies has shown that sensitive information can be reconstructed by training inversion networks~\cite{split_20} or by transmitting maliciously crafted gradients~\cite{split_22}, or \revise{through a combination of both approaches to enhance the attack~\cite{missing_2}. Moreover, \citet{missing_1} combined federated learning and split learning to improve privacy-preserving ability.}

A key difference~\cite{split_1,split_2} lies in the server’s visibility into the access to the client model, data distribution, or gradient information, creating a distinct threat landscape.
\revise{Recent work has demonstrated that client-side activations in split LLM inference are vulnerable to inversion attacks. \citet{DBLP:journals/corr/abs-2510-15511} proved that LLMs are almost-surely injective and proposed SipIt, which recovers prompts via autoregressive token-by-token search. \citet{DBLP:conf/ccs/Luo0X25} proposed an auxiliary model-based inversion method to predict inputs matching the given activation values, while \citet{DBLP:conf/uss/Dong00C0Z25} addressed the challenge of inverting deep-layer activations by optimizing over a compact set of vocabulary components rather than the full embedding space. \citet{DBLP:conf/sp/0004ZWXYLZ25} first recovered continuous input embeddings under distribution constraints and then mapped to discrete tokens via activation calibration combined with semantic speculation from an auxiliary language model. While these works collectively establish the feasibility of prompt inversion in split LLM inference, we show that a much simpler attack can also achieve state-of-the-art reconstruction accuracy, indicating that prompt inversion is easier than previously suggested. The practical benefit of \attack is its minimal assumptions about the attacker's capabilities and superior computational efficiency (See Section \ref{sec3}). On the defense side, common protection mechanisms include injecting random noise into activations \cite{missing_3} or applying activation sparsification. Another line of defense involves cryptographic techniques~\cite{split_2,abs-2409-19134}. However, such cryptographic approaches typically incur substantial computational overhead, rendering them impractical for deploying LLMs on resource-constrained devices.}

\section{Split Inference Protocol}
\label{sec2}

In split inference scenario, a client interacts with a server that hosts a proprietary LLM denoted as $F$, which is composed of $Q$ sequential blocks $\{L_j, \dots, L_Q\}$.
A common structure for each block $L_j$ comprises a multi-head attention layer $\text{MHA}_j(\cdot)$ and a feed-forward network layer $\text{FFN}_j(\cdot)$~\cite{llama3,qwen3}.
Formally, given an input hidden state \(\mathbf{h}_{j-1}\), the block computes $\mathbf{h}_j = L_j(\mathbf{h}_{j-1}) = \mathrm{FFN}_j \big(\mathrm{MHA}_j(\mathbf{h}_{j-1})\big) \in \mathbb{R}^{L \times D},$ where $L$ is the sequence length and $D$ is the hidden dimension.

For split inference~\cite{split_1}, \(F\) is partitioned into a client-side submodel \(F_C\) and a server-side submodel \(F_S\). The client holds the first \(Q_1\) blocks, and the server holds the remaining \(Q - Q_1\) blocks:
$
F_C = L_{Q_1} \circ \cdots \circ L_1, \ 
F_S = L_Q \circ \cdots \circ L_{Q_1+1}.
$
In practice, \(Q_1 \ll Q\), so that the client bears only a small fraction of the total computational cost.

Let $V$ be the vocabulary size. 
The client’s input is a token sequence $\mathbf{x} = \{x_1, \dots, x_L\}$, where each $x_i \in \{1, \dots, V\}$.
The LLM generates a response $\mathbf{x}_p$ autoregressively.
Starting from an empty $\mathbf{x}_p$, it repeatedly appends a new token predicted from the concatenation of the original prompt $\mathbf{x}$ and the already generated tokens.
This continues until an end-of-sequence token is produced or a maximum length is reached.
If $K$ tokens are generated, the final response is $\mathbf{x}_p = \{x_{L+1}, \dots, x_{L+K}\}$.
Since every prediction step follows the same split computation pattern, we focus on the first predicted token $x_{L+1}$ for clarity and omit $\mathbf{x}_p$ from the notation in what follows.

\textbf{Client-side computation.}
The client first maps the input tokens $\mathbf{x}$ into embeddings.
Let $\mathbf{E} \in \mathbb{R}^{V \times D}$ be the embedding matrix AND the embedding for token $x_i$ is the $x_i$-th row $\mathbf{e}_{x_i}$ of $\mathbf{E}$.
Thus, the initial hidden representation is $\mathbf{h}_0 = [\mathbf{e}_{x_1}, \cdots, \mathbf{e}_{x_L}]^\top \in \mathbb{R}^{L \times D}$.
The client then applies its local submodel $\mathbf{h}_{Q_1} = F_C(\mathbf{h}_0)$.
The resulting intermediate activations $\mathbf{h}_{Q_1}$ are transmitted to the server.
These activations are the only information about the client’s input that the server observes.

\textbf{Server-side computation.}
Upon receiving $\mathbf{h}_{Q_1}$, the server completes the forward pass by applying the remaining transformer blocks $\mathbf{h}_Q = F_S(\mathbf{h}_{Q_1}).$
Next, the server projects $\mathbf{h}_Q$ to the vocabulary space using an output projection matrix $\mathbf{E}_2 \in \mathbb{R}^{D \times V}$, producing the token logits $\mathbf{P} = \mathbf{h}_Q \mathbf{E}_2 \in \mathbb{R}^{L \times V}.$
A decoding strategy (e.g., greedy decoding or sampling) is then applied to $\mathbf{P}$ to select the next token $x_{L+1}$.

\section{Attack Design and Evaluation}
\label{sec3}

\subsection{Threat Model}
\label{sec_threat_model}

We consider an honest-but-curious server as the adversary.

\textbf{Attacker's goal.}
The server faithfully follows the split inference protocol but passively attempts to reconstruct the client’s original input sequence $\mathbf{x} = \{x_1, \dots, x_L\}$.
Since all subsequent tokens are generated by the server itself, we focus on the privacy of the initial user-supplied tokens.

\textbf{Attacker's knowledge and ability.}
The server has full white-box access to the LLM $F$ (including $\mathbf{E}$\footnote{In fact, $\mathbf{E}$ is also a component of the LLM. We decouple it from $F_C$ for ease of exposition.}) and enjoy sufficient computing resources.
The server can observe the client-transmitted intermediate activations.
The server is strictly prohibited from altering the decoded predictions $\{x_{L+1}, \ldots, x_{L+K}\}$ or engaging in any other active malicious actions to facilitate its attack.
The attack is thus purely passive and relies only on information that the server legitimately receives in standard split inference.

\subsection{Our Attack: \attack}

Note that the server has access to $\mathbf{h}_{Q_1}$, which contains information about the client's raw query $\mathbf{x}$.
A straightforward approach for the server to recover $\mathbf{x}$ would be to initialize a dummy query $\hat{\mathbf{x}}$ and iteratively refine it until its output from the client-side model closely matches the observed $\mathbf{h}_{Q_1}$.
However, this poses a significant challenge: $\hat{\mathbf{x}}$ is a discrete token sequence, not a continuous value.
While one could attempt to infer $\mathbf{x}$ by enumerating all potential candidate input sequences and comparing the distance between their generated activations and $\mathbf{h}_{Q_1}$, this brute-force method leads to an exponentially growing search space.
This rapidly becomes computationally infeasible for any realistic vocabulary size or query length.

\textbf{Attack scheme.}
To circumvent this combinatorial explosion, \attack relaxes the problem into the continuous embedding space.
Rather than directly guessing discrete tokens, the attack optimizes continuous embedding representations of $\hat{\mathbf{x}}$ until its forward activations through $F_C$ align with the observed $\mathbf{h}{Q_1}$.
Once optimized, $\hat{\mathbf{h}}_0^*$ is projected back into discrete tokens via nearest-neighbor search in $\mathbf{E}$.
This process unfolds in two key phases:
\begin{itemize}[leftmargin=*,topsep=1pt]
    \item \attack first randomly initializes a continuous embedding vector, $\hat{\mathbf{h}}_0 \in \mathbb{R}^{L \times D}$.
    Then $\hat{\mathbf{h}}_0$ is optimized to make $F_C(\hat{\mathbf{h}}_0)$ as close the observed intermediate activations $\mathbf{h}_{Q_1}$ as possible:
    \begin{equation}
    \label{eq_act_matching}
        \hat{\mathbf{h}}_0^* = \underset{\hat{\mathbf{h}}_0}{\arg \min} \  \text{Dist}\left(F_C(\hat{\mathbf{h}}_0), \ \mathbf{h}_{Q_1}\right).
    \end{equation}
    Notice that, since ${F}_C$ is composed entirely of differentiable layers~\cite{llm_survey}, gradient-based optimizers like Adam or SGD can be used effectively.
    The server's full knowledge of ${F}_C$'s parameters allows it to compute gradients efficiently.

    \item Once the optimization converges, the next step is to convert $\hat{\mathbf{h}}_0^*$ back into discrete tokens that form the reconstructed query.
    For each row $j$ of $\hat{\mathbf{h}}_0^*$ (which corresponds to the reconstructed embedding for the $j$-th token in $\mathbf{x}$), \attack identifies its closest word embedding within the predefined embedding layer $\mathbf{E}$.
    The token associated with this closest embedding then becomes the reconstructed token for position $j$.
    Formally, this process is expressed:
    $$
    \hat{x}_j = \underset{v \in \{1, \dots, V\}}{\arg \min} \text{Dist}(\hat{\mathbf{h}}_0^*[j], \mathbf{e}_v),
    $$
    where $\hat{\mathbf{h}}_0^*[j]$ refers to the $j$-th row of $\hat{\mathbf{h}}_0^*$ and $\mathbf{e}_v$ is the embedding of the token $v$.
    By repeating this projection for all $L$ token positions, \attack can reconstruct the entire input sequence $\hat{\mathbf{x}}$.
\end{itemize}
\revise{Although simple and akin to common inversion attacks, \attack is remarkably effective and avoids overly elaborate techniques, preserving its generality across split inference variants. 
Moreover, \citet{DBLP:journals/corr/abs-2510-15511} proved that LLMs are almost-surely injective, meaning that distinct input sequences must produce distinct hidden-layer activations.
The contrapositive of this injectivity result—identical activations imply identical inputs—provides a strong theoretical foundation for \attack: optimizing the continuous embeddings until the activation distance (Equation~\ref{eq_act_matching}) approaches zero should, under idealized conditions, guarantee perfect reconstruction of the entire sequence.
While their method, SipIt, leverages this injectivity for autoregressive, token-by-token recovery, \attack simultaneously optimizes all token embeddings.
Conceptually, this mirrors standard joint gradient descent versus coordinate descent.
This divergence grants \attack a notable advantage in computational efficiency.
Because the core complexity of the Transformer self-attention mechanism is $\mathcal{O}(n^2)$ for a sequence of length $n$, the total time complexity of \attack over $T$ optimization iterations is bounded by $\mathcal{O}(Tn^2)$.
By contrast, SipIt's sequential token reconstruction mandates repeated forward/backward passes of increasing lengths, resulting in a total computational cost of $\mathcal{O}(T(1^2 + \dots + n^2)) = \mathcal{O}(Tn^3)$.
Thus, \attack not only inherits the exact recovery guarantees provided by the injectivity theory but also translates them into a highly practical attack with a one-order-of-magnitude algorithmic speedup.}

\subsection{Evaluation Setup}

\textbf{Attack realization.}
In our specific implementation of \attack, the continuous embedding vector $\hat{\mathbf{h}}_0$ is initialized by randomly sampling $L$ rows directly from $\mathbf{E}$.
For $\text{Dist}(\cdot, \cdot)$, we adopt the cosine distance.
We use Adam optimizer with a learning rate of 0.01 to solve Equation \ref{eq_act_matching} for 2000 iterations.

\noindent
\textbf{Datasets.} 
We conduct our evaluations on two datasets: AlpacaEval and iCliniq.
AlpacaEval, with 805 records, serves as a benchmark for general knowledge question-answering, while iCliniq comprises 7321 dialogue records between patients and doctors, focusing on medical consultations.
Using the Qwen3 tokenizer, the average token lengths are 35.86 and 111.89 for AlpacaEval and iCliniq, respectively.

\begin{table}[!t]
\centering
\caption{\attack's reconstruction performance on AlpacaEval and iCliniq. We report Precision, Recall, and ROUGE-L scores against different sizes of Qwen3 and Falcon3 models. D. = Dataset. The values presented are the average results across samples, with the $\pm$ values indicating the standard deviation. Best results are in bold; second-best are underlined.}
\label{sec3_tab1}
\small
\begin{tabular}{c|l|ccc}
\hline \hline
{{Dataset}} & Model & Precision & Recall & ROUGE-L \\
\hline
\multirow{10}{*}{\rotatebox{90}{\makecell{AlpacaEval}}} 
    & Qwen3-0.6B  & \result{{99.83}}{0.61} & \result{98.54}{1.91} & \result{{0.96}}{0.04} \\
    & Qwen3-1.7B  & \result{\textbf{100.00}}{0.02}   & \result{\textbf{99.76}}{0.91} & \result{\textbf{0.99}}{0.02} \\
    & Qwen3-4B    & \result{99.77}{1.01}             & \result{{99.59}}{1.24} & \result{0.99}{0.02} \\
    & Qwen3-8B    & \result{\underline{99.91}}{1.98}             & \result{99.31}{1.52} & \result{0.98}{0.03} \\
    & Qwen3-14B   & \result{99.89}{1.70}             & \result{\underline{99.62}}{1.28} & \result{0.98}{0.02} \\
    & Qwen3-30B   & \result{99.95}{1.59}             & \result{99.49}{1.73} & \result{\underline{0.99}}{0.04} \\
    & Falcon3-1B  & \result{95.69}{2.76}             & \result{92.25}{4.28} & \result{0.90}{0.06} \\
    & Falcon3-3B  & \result{97.89}{1.96}             & \result{95.83}{2.45} & \result{0.95}{0.04} \\
    & Falcon-7B   & \result{98.82}{2.31}             & \result{97.90}{3.70} & \result{0.97}{0.06} \\
    & Falcon-10B  & \result{99.25}{2.21}             & \result{99.14}{3.21} & \result{0.99}{0.05} \\
\hline
\multirow{10}{*}{\rotatebox{90}{iCliniq}} 
    & Qwen3-0.6B  & \result{99.81}{0.42}             & \result{98.45}{1.33} & \result{{0.98}}{0.02} \\
    & Qwen3-1.7B  & \result{\textbf{99.99}}{0.06}    & \result{\textbf{99.90}}{0.33} & \result{\textbf{0.99}}{0.01} \\
    & Qwen3-4B    & \result{\underline{99.94}}{0.27} & \result{\underline{99.84}}{0.44} & \result{\underline{0.99}}{0.01} \\
    & Qwen3-8B    & \result{99.82}{1.79}             & \result{99.28}{1.65} & \result{0.98}{0.02} \\
    & Qwen3-14B   & \result{99.75}{1.05}             & \result{99.54}{1.81} & \result{0.99}{0.02} \\
    & Qwen3-30B   & \result{99.90}{1.29}             & \result{99.47}{1.25} & \result{0.99}{0.03} \\
    & Falcon3-1B  & \result{96.81}{2.39}             & \result{92.10}{2.77} & \result{0.91}{0.04} \\
    & Falcon3-3B  & \result{98.48}{1.27}             & \result{96.17}{1.90} & \result{0.96}{0.03} \\
    & Falcon-7B   & \result{98.86}{2.28}             & \result{97.58}{3.26} & \result{0.98}{0.04} \\
    & Falcon-10B  & \result{99.20}{2.19}             & \result{99.12}{2.08} & \result{0.99}{0.05} \\
\hline \hline
\end{tabular}
\end{table}

\noindent
\textbf{Models.}
We employ two recently released LLM families: Qwen3-\{0.6, 1.7, 4, 8, 14, 30\}B~\cite{qwen3} and Falcon3-\{1, 3, 7, 10\}B\footnote{\url{https://huggingface.co/blog/falcon3}}.

\noindent
\textbf{Metrics.}
We include Precision, Recall, and ROUGE-L to evaluate the similarity between the reconstructed query and the original query.
Precision measures the proportion of correctly reconstructed tokens among all tokens in the reconstructed query while Recall measures the proportion of correctly reconstructed tokens among all tokens in the original query.
ROUGE-L assesses the overlap between the reconstructed query and the original query based on their longest common subsequence.

\noindent
\textbf{Hyperparameters.}
For the client-side model configuration, we use $Q_1=5$, meaning the client-side model ($F_C$) consists of the first five blocks of the LLM.

\subsection{Evaluation Results}

\begin{table*}[!t]
\centering
\caption{The \attack's effectiveness against random Gaussian noise and activation sparsification (element and token) in Qwen3-0.6B and Falcon3-1B. 
The defense levels, ranging from 1 to 5, correspond to Gaussian noise intensities $\{10^{-4},10^{-3},10^{-2},10^{-1},1\}$ 
and sparsity ratios $\{0.1,0.3,0.5,0.7,0.9\}$. M. = Model; L. = Level. Results are presented as mean $\pm$ std., with the best results highlighted in bold.}
\label{sec3_tab2}
\small
\begin{tabular}{c|l|c|ccc|ccc|ccc}
\hline \hline
\multirow{2}{*}{Dataset} & \multirow{2}{*}{M.} & \multirow{2}{*}{L.}
& \multicolumn{3}{c|}{Precision} 
& \multicolumn{3}{c|}{Recall} 
& \multicolumn{3}{c}{ROUGE-L} \\
& & & Gaussian & Element & Token 
      & Gaussian & Element & Token 
      & Gaussian & Element & Token \\ \hline

\multirow{10}{*}{\rotatebox{90}{AlpacaEval}}
& \multirow{5}{*}{\rotatebox{90}{Qwen3-0.6B}} 
& 1 & \result{\textbf{99.81}}{0.62} & \result{99.68}{0.88} & \result{99.73}{0.96} 
    & \result{98.49}{2.34} & \result{98.38}{2.42} & \result{\textbf{98.56}}{2.06} 
    & \result{0.95}{0.03}  & \result{\textbf{0.96}}{0.03}  & \result{{0.96}}{0.03} \\
& & 2 & \result{\textbf{99.73}}{0.78} & \result{99.38}{1.37} & \result{99.39}{1.76} 
    & \result{\textbf{98.51}}{2.40} & \result{97.41}{2.66} & \result{97.04}{2.99} 
    & \result{\textbf{0.96}}{0.04}  & \result{0.95}{0.03}  & \result{0.95}{0.04} \\
&  & 3 & \result{\textbf{99.39}}{1.91} & \result{92.92}{7.23} & \result{95.65}{5.09} 
    & \result{\textbf{98.34}}{2.41} & \result{85.53}{7.54} & \result{89.16}{6.06} 
    & \result{\textbf{0.95}}{0.04}  & \result{0.93}{0.04}  & \result{0.93}{0.04} \\
&  & 4 & \result{\textbf{94.53}}{4.44} & \result{69.63}{10.96} & \result{72.85}{10.39} 
    & \result{\textbf{86.77}}{5.80} & \result{61.05}{9.38}  & \result{63.72}{8.85} 
    & \result{\textbf{0.93}}{0.04}  & \result{0.79}{0.07}   & \result{0.79}{0.08} \\
&  & 5 & \result{3.63}{3.45}  & \result{\textbf{27.80}}{9.78}  & \result{25.89}{10.78} 
    & \result{2.94}{2.62}  & \result{\textbf{22.57}}{7.40}  & \result{21.04}{7.41} 
    & \result{0.06}{0.05}  & \result{\textbf{0.35}}{0.10}   & \result{0.29}{0.10} \\ \cline{2-12}

& \multirow{5}{*}{\rotatebox{90}{Falcon3-1B}} 
& 1 & \result{\textbf{96.36}}{2.74} & \result{95.63}{2.63} & \result{95.46}{2.37}
    & \result{92.71}{3.97} & \result{\textbf{92.78}}{4.65} & \result{92.53}{4.06}
    & \result{\textbf{0.91}}{0.06}  & \result{{0.91}}{0.07}  & \result{0.90}{0.07} \\
&  & 2 & \result{\textbf{95.84}}{3.20} & \result{95.42}{3.00} & \result{95.70}{2.71}
    & \result{\textbf{93.19}}{4.01} & \result{92.46}{4.98} & \result{91.99}{4.02}
    & \result{\textbf{0.91}}{0.06}  & \result{{0.91}}{0.07}  & \result{0.90}{0.07} \\
&  & 3 & \result{\textbf{95.76}}{3.22} & \result{94.26}{3.05} & \result{95.05}{3.12}
    & \result{\textbf{93.05}}{4.21} & \result{91.02}{4.46} & \result{91.47}{4.28}
    & \result{\textbf{0.92}}{0.06}  & \result{0.90}{0.06}  & \result{0.90}{0.06} \\
&  & 4 & \result{\textbf{95.57}}{3.24} & \result{88.38}{5.53} & \result{89.89}{5.36}
    & \result{\textbf{93.08}}{4.19} & \result{81.62}{6.73} & \result{84.44}{6.27}
    & \result{\textbf{0.92}}{0.06}  & \result{0.84}{0.07}  & \result{0.84}{0.08} \\
&  & 5 & \result{\textbf{78.72}}{4.96} & \result{34.41}{9.39} & \result{41.84}{11.62}
    & \result{\textbf{67.95}}{8.58} & \result{25.53}{5.94} & \result{30.97}{7.43}
    & \result{\textbf{0.65}}{0.10}  & \result{0.23}{0.10}  & \result{0.27}{0.12} \\ \hline

\multirow{10}{*}{\rotatebox{90}{iCliniq}}
&\multirow{5}{*}{\rotatebox{90}{Qwen3-0.6B}} 
& 1 & \result{\textbf{99.72}}{0.60} & \result{99.68}{0.69} & \result{99.69}{0.58} 
    & \result{\textbf{98.57}}{1.84} & \result{98.37}{1.55} & \result{98.44}{1.25} 
    & \result{{0.98}}{0.02}  & \result{\textbf{0.98}}{0.02}  & \result{{0.98}}{0.02} \\
& & 2 & \result{99.69}{0.56} & \result{99.38}{0.97} & \result{\textbf{99.73}}{0.49} 
    & \result{\textbf{98.72}}{1.32} & \result{96.94}{1.93} & \result{97.87}{1.51} 
    & \result{\textbf{0.98}}{0.01}  & \result{{0.98}}{0.02}  & \result{{0.98}}{0.02} \\
&  & 3 & \result{\textbf{99.67}}{0.73} & \result{96.91}{3.55} & \result{97.71}{1.83} 
    & \result{\textbf{98.48}}{1.76} & \result{91.18}{3.95} & \result{92.72}{3.19} 
    & \result{\textbf{0.98}}{0.01}  & \result{0.96}{0.03}  & \result{0.96}{0.02} \\
&   & 4 & \result{\textbf{96.83}}{2.85} & \result{84.29}{8.05} & \result{84.25}{8.15} 
    & \result{\textbf{90.75}}{4.28} & \result{73.25}{7.45} & \result{74.44}{7.30} 
    & \result{\textbf{0.95}}{0.03}  & \result{0.85}{0.06}  & \result{0.86}{0.06} \\
& & 5 & \result{4.08}{3.08}  & \result{40.83}{11.28} & \result{\textbf{41.16}}{11.43} 
    & \result{2.66}{1.95}  & \result{28.86}{6.60}  & \result{\textbf{28.87}}{7.11} 
    & \result{0.05}{0.03}  & \result{\textbf{0.37}}{0.07}   & \result{0.35}{0.08} \\ \cline{2-12}

&\multirow{5}{*}{\rotatebox{90}{Falcon3-1B}}
& 1 & \result{96.82}{1.93} & \result{96.94}{2.13} & \result{\textbf{96.96}}{2.20}
    & \result{92.10}{2.88} & \result{91.76}{3.08} & \result{\textbf{92.14}}{3.08}
    & \result{\textbf{0.91}}{0.04}  & \result{0.90}{0.04}  & \result{0.90}{0.05} \\
&  & 2 & \result{96.91}{1.92} & \result{96.89}{2.10} & \result{\textbf{97.11}}{1.99}
    & \result{\textbf{92.25}}{2.93} & \result{91.81}{3.42} & \result{91.73}{2.92}
    & \result{\textbf{0.91}}{0.04}  & \result{0.90}{0.04}  & \result{{0.91}}{0.04} \\
&   & 3 & \result{\textbf{96.89}}{1.78} & \result{96.26}{2.63} & \result{96.51}{2.05}
    & \result{\textbf{92.41}}{2.87} & \result{91.28}{2.67} & \result{90.80}{3.16}
    & \result{\textbf{0.91}}{0.04}  & \result{0.90}{0.04}  & \result{0.90}{0.04} \\
&   & 4 & \result{\textbf{97.10}}{1.99} & \result{93.48}{2.97} & \result{93.90}{2.08}
    & \result{\textbf{92.65}}{3.10} & \result{85.94}{3.47} & \result{85.87}{3.74}
    & \result{\textbf{0.91}}{0.04}  & \result{0.85}{0.06}  & \result{0.85}{0.04} \\
&   & 5 & \result{\textbf{81.24}}{5.82} & \result{48.12}{9.85} & \result{52.27}{8.34}
    & \result{\textbf{68.32}}{4.92} & \result{34.46}{5.21} & \result{39.89}{5.68}
    & \result{\textbf{0.66}}{0.06}  & \result{0.31}{0.06}  & \result{0.39}{0.07} \\ \hline \hline
\end{tabular} 
\end{table*}

\begin{figure*}[!ht]
\centering
\includegraphics[width=1\linewidth]{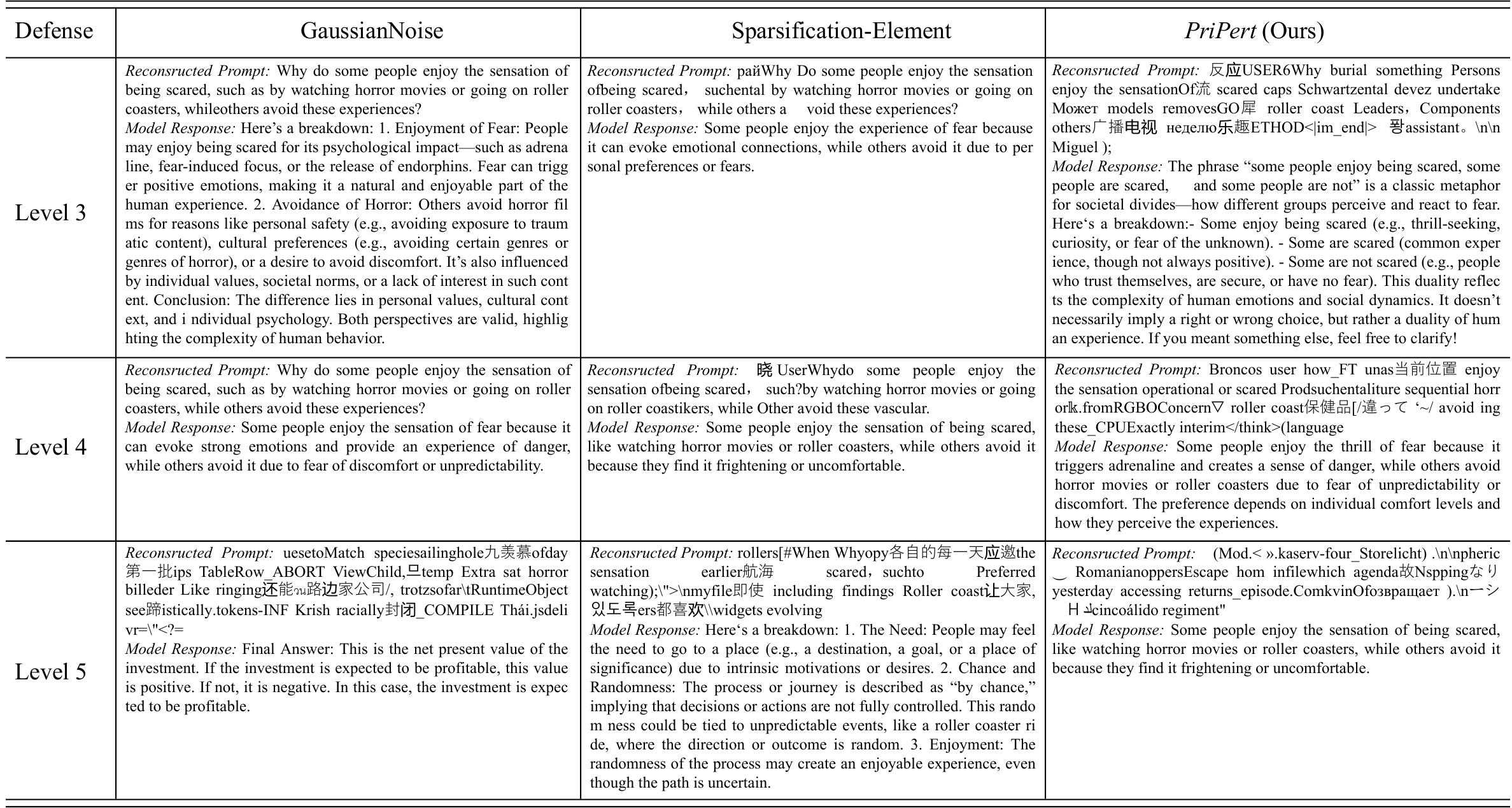} 
\caption{We randomly extract a prompt from AlpacaEval: \textcolor{cyan}{Why do some people enjoy the sensation of being scared, such as by watching horror movies or going on roller coasters, while others avoid these experiences?} We apply \attack to reconstruct the original prompt using activations perturbed by different defense methods. The corresponding model responses using perturbed intermediate activations are also provided.}
\label{fig_case}
\end{figure*}

\begin{table}[!t]
\centering
\caption{The impact of increasing the number of client-side blocks ($Q_1$) on \attack's performance. We employ Qwen3-0.6B. The values listed are the mean results with their corresponding standard deviations. The best results are given in bold.}
\label{sec3_tab3}
\small
\begin{tabular}{c|c|ccc}
\hline \hline
{Dataset}              & Block & Precision & Recall & ROUGE-L \\ \hline
\multirow{6}{*}{\rotatebox{90}{{AlpacaEval}}} 
    & 2 & \result{\textbf{99.76}}{1.20} & \result{\textbf{98.89}}{1.94} & \result{\textbf{1.00}}{0.00} \\ 
    & 3 & \result{99.34}{1.50} & \result{96.95}{2.90} & \result{{1.00}}{0.01} \\
    & 4 & \result{95.79}{5.19} & \result{91.05}{6.46} & \result{0.94}{0.04} \\
    & 5 & \result{92.92}{7.23} & \result{85.53}{7.54} & \result{0.93}{0.04} \\
    & 6 & \result{86.09}{8.10} & \result{77.05}{8.23} & \result{0.91}{0.05} \\
    & 7 & \result{77.74}{9.66} & \result{69.82}{8.30} & \result{0.87}{0.06} \\ \hline
\multirow{6}{*}{\rotatebox{90}{iCliniq}} 
    & 2 & \result{\textbf{99.91}}{0.32} & \result{\textbf{99.60}}{1.00} & \result{\textbf{1.00}}{0.01} \\
    & 3 & \result{99.58}{0.72} & \result{98.48}{1.31} & \result{{1.00}}{0.01} \\
    & 4 & \result{98.38}{2.05} & \result{95.33}{2.74} & \result{0.98}{0.02} \\
    & 5 & \result{96.91}{3.55} & \result{91.18}{3.95} & \result{0.96}{0.03} \\
    & 6 & \result{93.76}{3.95} & \result{85.20}{4.61} & \result{0.94}{0.03} \\
    & 7 & \result{86.81}{6.41} & \result{75.81}{5.25} & \result{0.89}{0.05} \\ \hline \hline
\end{tabular}
\end{table}

\begin{figure*}[!ht]
    \centering
    \subfloat[Precision\label{fig:precision}]{
        \includegraphics[width=0.32\linewidth]{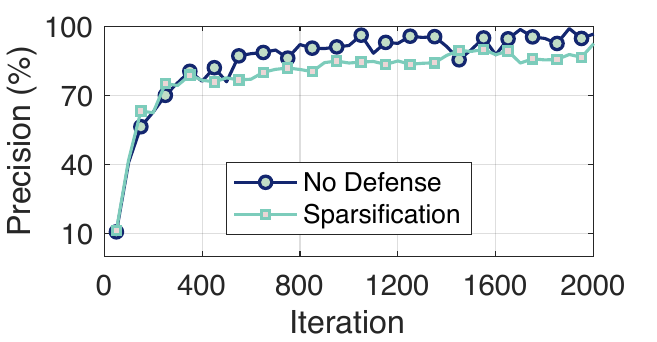}
    }
    \hfill
    \subfloat[Recall\label{fig:recall}]{
        \includegraphics[width=0.32\linewidth]{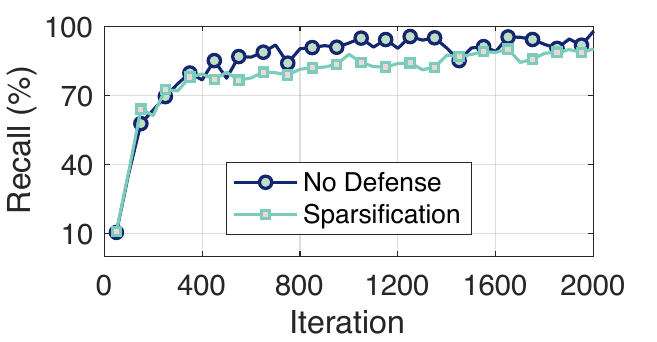}
    }
    \hfill
    \subfloat[ROUGE-L\label{fig:recall}]{
        \includegraphics[width=0.32\linewidth]{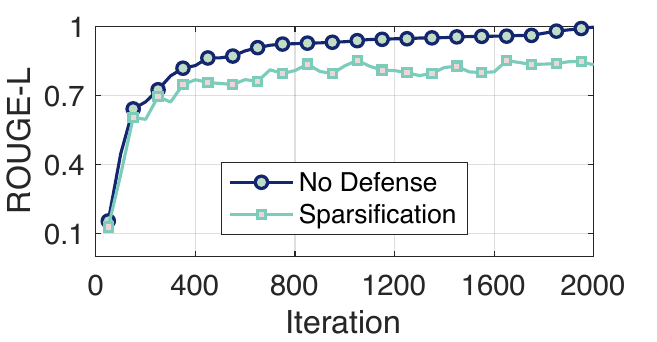}
    }
    \caption{The \attack's Precision, Recall, and ROUGE-L scores evolve over 2000 optimization iterations in AlpacaEval. We use a sparsification ratio of 0.5.}
    \label{sec3_fig1}
\end{figure*}

\textbf{Overall performance.}
Table \ref{sec3_tab1} reports the performance of \attack against various LLMs on both AlpacaEval and iCliniq.
The results demonstrate that \attack consistently achieves Precision and Recall exceeding 98\% across nearly all evaluated cases.
Furthermore, the ROUGE-L scores remain consistently high, typically surpassing 0.96.
These results confirm that, under our honest-but-curious threat model, vanilla split inference provides little practical privacy, which motivates our examination into whether perturbing activations can effectively defend against \attack.

\noindent
\textbf{Evaluating common defense methods.}
We now evaluate the effectiveness of two common defense strategies against \attack: activation sparsification and random noise perturbation.
Since we observe that larger models tend to exhibit stronger privacy leakage, we report results for Qwen3-0.6B and Falcon3-1B here to better illustrate how different factors relate to privacy leakage.
Please see Table~\ref{sec3_tab_qwen30b} for the results in larger models.
Activation sparsification reduces privacy leakage by zeroing out a percentage of the intermediate activations with the smallest absolute magnitude, while another one is to perturb the intermediate activations with random Gaussian noise.
Activation sparsification can be applied at different granularities, namely at the element level or at the token level.
Both defenses can obfuscate the true activations, making reconstruction more difficult.
Notice that we only perturb exclusively the activations of the client-side model's final block, rather than all blocks.
Table~\ref{sec3_tab2} shows that neither method substantially impairs \attack.
For activation sparsification, we vary the pruning ratio in \(\{0.1, 0.3, 0.5, 0.7, 0.9\}\).
Even when pruning up to 70\% of elements, Precision and Recall decrease only modestly.
\textbf{Element-level sparsification is more effective than token-level sparsification, so we adopt element-level sparsification as the default sparsification defense in the remainder of the paper.}
For Gaussian noise, we consider variances in \(\{10^{-4}, 10^{-3}, 10^{-2}, 10^{-1}\}\).
Across this range, \attack remains largely resilient.
Only when the noise magnitude reaches \(10^{0}\) do we see a substantial drop in attack performance, but at that point, the downstream inference accuracy of the server-side model collapses as well.
This illustrates a fundamental tension: noise strong enough to meaningfully protect privacy also destroys utility.
We next provide qualitative examples and leave the detailed quantitative results to Section \ref{sec5}.

\noindent
\textbf{Case study.}
We select a sample from AlpacaEval to illustrate \attack's performance against varying defenses.
As seen in Figure \ref{fig_case}, with defense levels 3-4, the reconstructed queries are nearly indistinguishable from the originals, demonstrating high-fidelity reconstruction with only minor lexical artifacts.
This indicates that \attack not only recovers token-level details but also accurately captures the semantic intent behind the query.
At the strongest defense setting (level 5), reconstruction quality degrades markedly, with missing and corrupted segments.
Correspondingly, the model's responses also become corrupted.

\noindent
\textbf{Number of blocks ($Q_1$).}
Table~\ref{sec3_tab3} reports the performance of \attack when varying the number of client-side blocks ($Q_1$) from 2 to 7.
To avoid confounding effects due to excessively strong attack performance, we adopt activation sparsification (element-level) with a fixed ratio of 0.5 by default.
As $Q_1$ increases, reconstruction quality gradually deteriorates.
This is likely due to two factors: first, each subsequent block performs additional feature extraction and compression, leading to an inherent loss of information that hinders the reconstruction process.
Second, as the client-side model becomes more complex, the reconstruction optimization problem becomes more challenging.
However, this effect is limited.
Even with a large $Q_1 = 7$, \attack still achieves Precision and Recall around or above 70\%.
This indicates that while it offers a degree of mitigation, it does not completely deter \attack.

\noindent
\textbf{Attack cost: number of iterations.}
Figure \ref{sec3_fig1} shows the change in Precision, Recall, and ROUGE-L over the course of the attack, averaged over 100 samples.
As shown, most of the performance gain is achieved within the first 1000 iterations, at which point most semantics of the client's inputs can be inferred, with minimal improvement thereafter.
We then time the attack on a single NVIDIA GeForce RTX 4090 GPU and find that 100 iterations per token require $0.2012 \pm 0.1095$ seconds (mean $\pm$ standard deviation).
This indicates that for a 100-token query, \attack would take approximately 20 seconds for 1000 iterations and 40 seconds for 2000 iterations.
These times are not a significant practical constraint for an attacker.

\noindent
\revise{\textbf{Comparison with state-of-the-art attacks.}
We here compare \attack against four prompt inversion attacks, including A1~\cite{DBLP:conf/ccs/Luo0X25}, SipIt~\cite{DBLP:journals/corr/abs-2510-15511}, TBS~\cite{DBLP:conf/uss/Dong00C0Z25}, and PIA~\cite{DBLP:conf/sp/0004ZWXYLZ25}.
Concretely, SipIt performs autoregressive token-by-token inversion, while A1 trains a surrogate model to predict the token corresponding to a given hidden state.
TBS and PIA are more closely related to \attack in that they also attempt to recover hidden embeddings and then project them back to the discrete token space.
The difference is that TBS optimizes over a restricted vocabulary subspace, whereas PIA leverages an auxiliary model to guide the projection step.
In contrast, \attack directly optimizes the entire prompt embedding sequence to match the observed client-side activations, and projects the optimized embeddings back to tokens only at the final step via nearest-neighbor search.
We use Qwen3-8B, AlpacaEval, and $Q_1=5$.
Because SipIt incurs substantially higher runtime, we truncate each prompt to 10\% of its original length when comparing SipIt with \attack.
Table~\ref{tab:attack_compare_all} yields three observations where we use element-level sparsification with a ratio of 0.5 and Gaussian noise with magnitude $10^{-1}$.
First, \attack achieves reconstruction quality on par with the strongest prior attacks.
Its reconstruction score reaches about 99\%, essentially matching SipIt, slightly exceeding PIA, and substantially outperforming TBS and A1.
This is particularly notable given that \attack does not rely on additional resources such as auxiliary models.
Second, \attack is substantially more efficient in practice.
Among the baselines, SipIt is the closest competitor in reconstruction quality, but it requires 4.92 minutes per sample, whereas \attack takes only 0.78 minutes.
This efficiency gain arises because SipIt reconstructs prompts sequentially, while \attack jointly optimizes all token embeddings in a single continuous optimization process.
We note that A1 appears very efficient because most of its cost lies in collecting data and training a lightweight surrogate model.
Here, we report only the inference-time cost of the trained surrogate.
Overall, if the remaining baselines are also considered, \attack offers a favorable trade-off between reconstruction quality and attack cost.}

\begin{table}[t]
\centering
\caption{\revise{Comparison between \attack and four attacks on Qwen3-8B. We report the arithmetic mean of token-level Precision and Recall and the average runtime per sample on a single RTX 4090.}}
\label{tab:attack_compare_all}
\small
\begin{tabular}{l|c|c|c}
\hline \hline
\textbf{Attack} & \textbf{No Defense} & \textbf{Sparsification} & \textbf{Time / sample} \\
\hline
A1~\cite{DBLP:conf/ccs/Luo0X25} & 43.08 & 18.57 & 0.22 min \\
TBS~\cite{DBLP:conf/uss/Dong00C0Z25} & 88.32 & 81.97 & 7.43 min \\
PIA~\cite{DBLP:conf/sp/0004ZWXYLZ25} & 98.48 & 95.57 & 9.59 min \\
\attack & \textbf{99.52} & \textbf{96.61} & \textbf{7.20 min} \\
\hline
\hline
\textbf{Attack} & \textbf{No Defense} & \textbf{Gaussian Noise} & \textbf{Time / sample} \\
\hline
SipIt~\cite{DBLP:journals/corr/abs-2510-15511} & \textbf{99.69} & 90.63 & 4.92 min \\
\attack & 99.65 & \textbf{90.65} & \textbf{0.78 min} \\
\hline \hline
\end{tabular}
\end{table}

\section{Unpacking Attack Effectiveness: What Makes \attack So Potent?}
\label{sec4}

Section \ref{sec3} demonstrates that \attack is surprisingly effective at reconstructing private inputs, even under common perturbation defenses.
This finding motivates a deeper investigation into the underlying factors that contribute to \attack's success.
Specifically, this section aims to study: \textbf{\textit{What components of LLMs render them so vulnerable to \attack?}}

To address this, we define a layer-wise sensitivity metric to quantify a layer's intrinsic resistance to reconstruction.
\revise{Intuitively, a high-sensitivity layer is one where small changes at its output induce large changes at its input.
In such layers, even tiny perturbations at the output can severely disrupt inversion, making accurate recovery of the original input difficult. Conversely, low-sensitivity layers are locally stable.
Wherein, small output discrepancies translate into small input errors, so approximate activation matching, as performed by \attack, is sufficient to recover the underlying embeddings and tokens.
By understanding and quantifying this layer-wise sensitivity, we can identify vulnerable components and pave the way for designing more effective defenses.}

\subsection{Quantifying Layer Sensitivity}

Let $\mathbf{y} = f(\mathbf{z}), \mathbf{z} \in \mathbb{R}^{1 \times n},  \mathbf{y} \in \mathbb{R}^{1 \times m}$ represent the forward pass of a certain layer.
To assess the sensitivity of this layer, we analyze how a perturbation $\delta$ introduced at the output, $\mathbf{y}$, affects the corresponding reconstructed input, $\hat{\mathbf{z}}$:
\begin{equation}
\label{sec4_1_eq1}
    \mathbf{y} + \delta = f(\hat{\mathbf{z}}).
\end{equation}
A layer is considered highly sensitive if a small $\delta$ causes a substantial deviation in the reconstructed input, i.e., $\|\hat{\mathbf{z}} - \mathbf{z}\|_2$ is large.
This characteristic is desirable from a defense perspective, as it makes it more difficult for an adversary to accurately reconstruct the original input.
To analyze the relationship between the output perturbation $\delta$ and the input deviation $\Delta = \hat{\mathbf{z}} - \mathbf{z}$, we apply a first-order Taylor expansion of $f$ around $\mathbf{z}$:
$f(\hat{\mathbf{z}}) = f(\mathbf{z} + \Delta) \approx f(\mathbf{z}) + \Delta \mathbf{J},$
where $\mathbf{J} \in \mathbb{R}^{n \times m}$ is the Jacobian matrix of $f$ evaluated at $\mathbf{z}$.
This approximation is justified because most LLM layers are linear transformations, and nonlinearities are locally linear over typical operating ranges.
Substituting this approximation back into Equation \ref{sec4_1_eq1} gives:
$\mathbf{y} + \delta \approx f(\mathbf{z}) +  \Delta \mathbf{J}.$
Since $\mathbf{y} = f(\mathbf{z})$, we can simplify this to: $\delta \approx  \Delta \mathbf{J},$ with the relationship between $\delta$ and $\Delta$ following Theorem \ref{thm_linear}.

\begin{theorem}
\label{thm_linear}
    Consider $\delta =  \Delta \mathbf{J}.$
    Let $\{\sigma_i\}_{i=1}^n$ and $\{\mu_i\}_{i=1}^n$ denote the eigenvalues and orthonormal eigenvectors of $\mathbf{J} \mathbf{J}^\top,$ respectively.
    Then, we have:
    \begin{equation}
    \label{sec4_1_eq2}
        ||\Delta||_2^2 \leq \sum_i \left( \frac{||\delta \mathbf{J}||_2 \cos \theta_i }{\sqrt{\sigma_i+1 }} \right)^2,
    \end{equation}
    where $\theta_i$ denotes the angle between $\delta \mathbf{J}$ and the subspace spanned by $\mu_i$.
\end{theorem}

\begin{figure}[!t]
    \centering
    \includegraphics[width=0.7\linewidth]{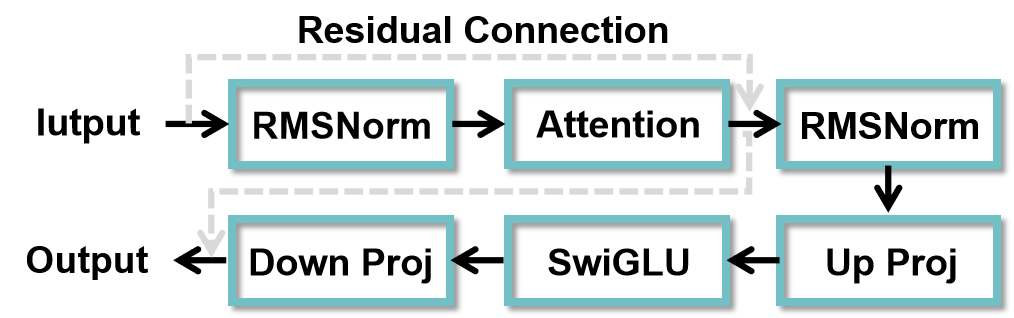}
    \caption{The common components of a single block within modern LLMs.}
    \label{sec4_fig1}
\end{figure}

\begin{figure*}[!ht]
    \centering
    \subfloat[Expected PAF\label{fig_qwen_avg}]{
        \includegraphics[width=0.49\linewidth]{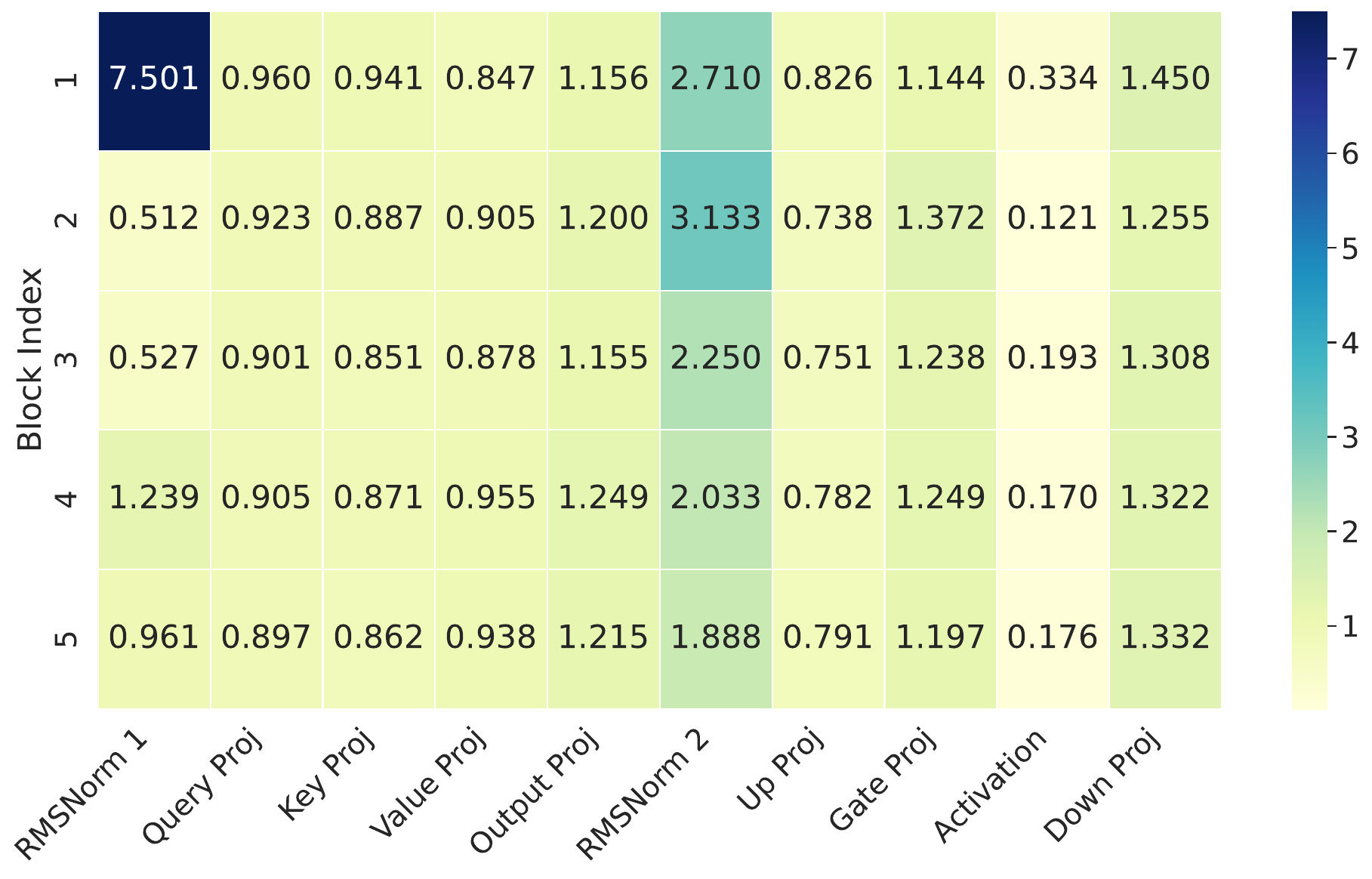}
    }
    \hfill
    \subfloat[Max-PAF\label{fig_qwen_worst}]{
        \includegraphics[width=0.49\linewidth]{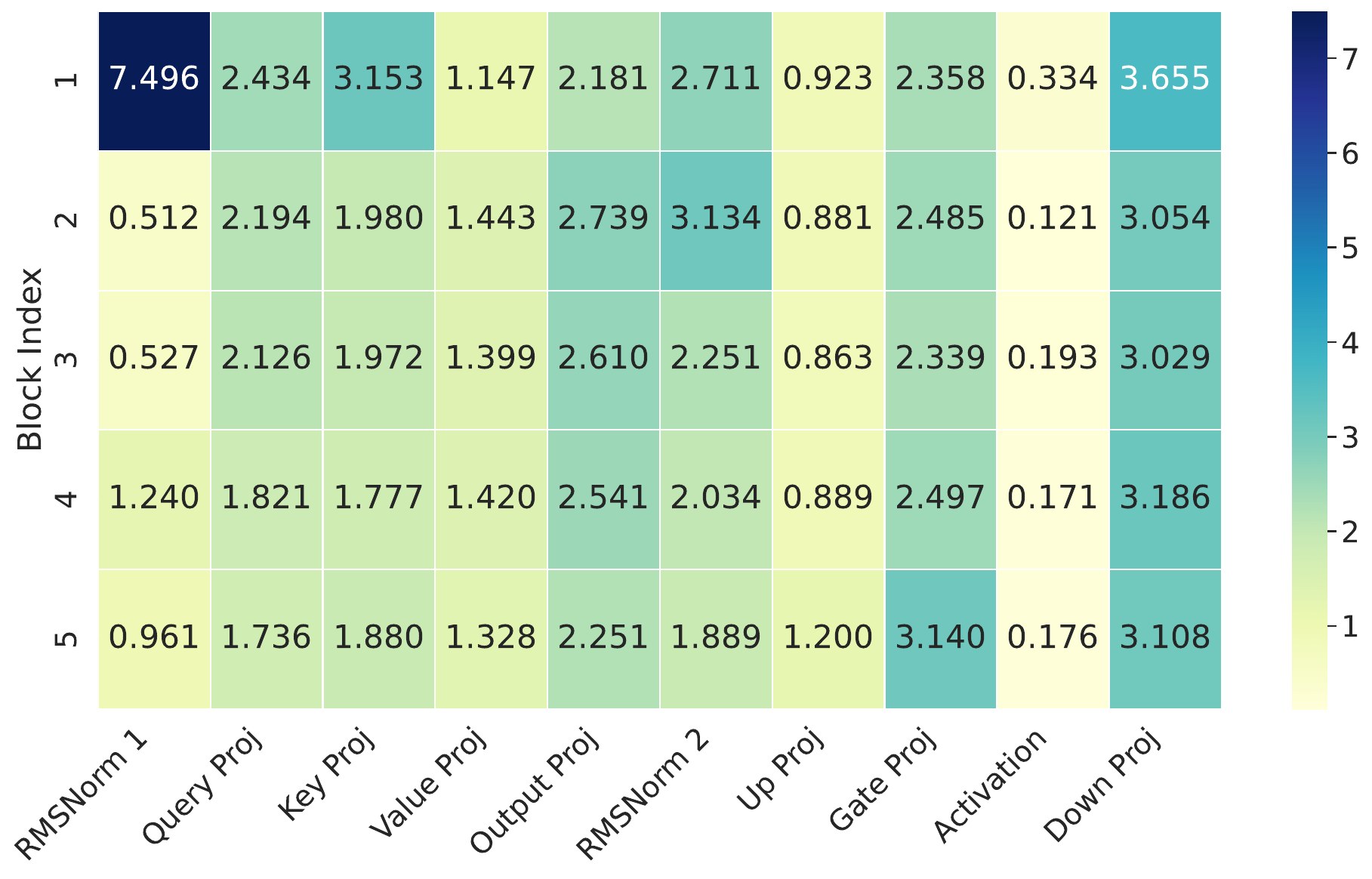}
    }
    \caption{Comparison of different layers' sensitivity in Qwen3-0.6B. The expected PAF values capture the average amplification across random perturbations, and highest PAF represents the maximum possible amplification.}
    \label{fig_qwen}
\end{figure*}

\begin{figure*}[!ht]
    \centering
    \subfloat[Expected PAF\label{fig_fal_avg}]{
        \includegraphics[width=0.49\linewidth]{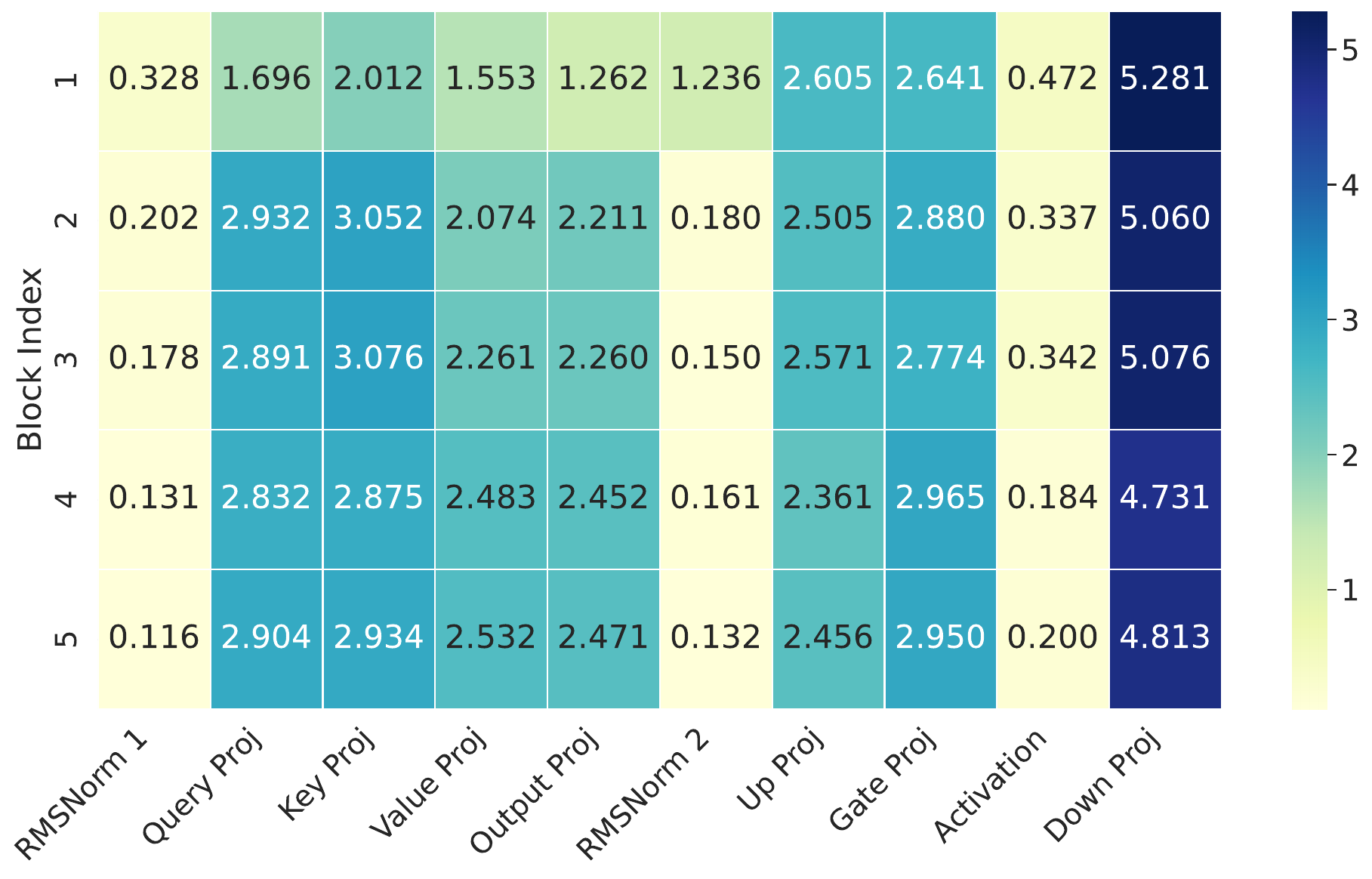}
    }
    \hfill
    \subfloat[Max-PAF\label{fig_fal_worst}]{
        \includegraphics[width=0.49\linewidth]{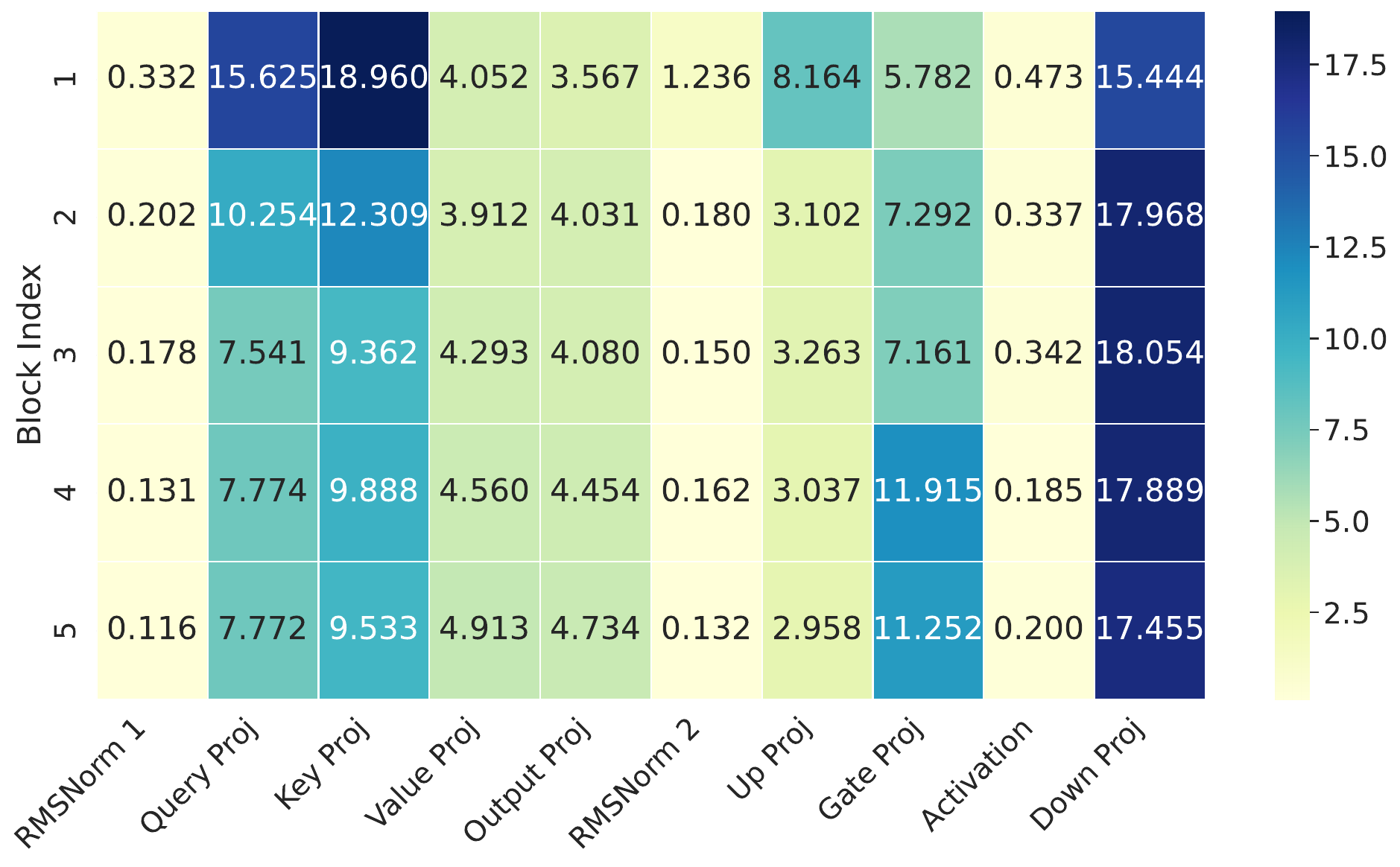}
    }
    \caption{Comparison of different layers' sensitivity in Falcon3-1B.}
    \label{fig_fal}
\end{figure*}

The proof of Theorem \ref{thm_linear} is provided in Appendix \ref{proof1}.
Leveraging Equation \ref{sec4_1_eq2}, we can define a quantitative measure for a layer's sensitivity.
Applying Cauchy–Schwarz inequality and dividing by $\|\delta\|_2^2$ yields:
$
\frac{||\Delta||_2^2}{||\delta||_2^2} \leq \sum_{i=1}^n \left( \frac{||\mathbf{J}||_2 \cos \theta_i }{\sqrt{\sigma_i+1}} \right)^2.
$
In light of this, we define Perturbation Amplification Factor (PAF) as a metric for a layer's sensitivity:\footnote{We opt for the absolute value instead of the squared value due to its better numerical stability observed during our experiments.}:
\begin{equation}
\label{eq_PAF}
    \text{PAF} := \mathbb{E}_{\delta \sim \mathcal{N}(\mathbf{0}, \mathbf{I})} \sum_i \left| \frac{|| \mathbf{J}||_2 \cos \theta_i }{\sqrt{\sigma_i+1 }}  \right|,
\end{equation}
where the expectation is taken over $\delta$ drawn from an isotropic Gaussian distribution $\mathcal{N}(\mathbf{0}, \mathbf{I})$.
A higher PAF value indicates that output perturbations are more effective at causing substantial input deviations, signifying a greater intrinsic resistance to reconstruction.
The next subsection empirically evaluates the PAF values for different layers within an LLM.

\subsection{Empirical Evaluation}
\label{sec_4_2}

\textbf{Setup.}
We employ four LLMs, including Qwen3-0.6B, Falcon3-1B, Llama-3.2-1B~\cite{llama3}, and SmolLM2-1.7B~\cite{SmolLM2}, and evaluate their PAF values for distinct layers.
Figure \ref{sec4_fig1} illustrates the main components of a Qwen3 block, which consists of two RMSNorm layers, an MHA layer, and an FFN layer.
The MHA sub-layer is further composed of query, key, value, and output projection layers, while the FFN layer comprises up-projection, gate projection, sigmoid activation, and down-projection layers\footnote{Gate Projection and sigmoid function together form SwiGLU activation function, a standard component in modern LLMs. While SwiGLU is the de facto activation, a common convention in the LLM research community is to analyze the Gate Projection and sigmoid component separately, often referring to the latter as Activation. We adhere to this established practice.}.
Residual connections are excluded from this analysis, as their identity mapping inherently yields a sensitivity of 1.
Because $\mathbf{J}$ is input-dependent, we evaluate it on a random subset of 100 AlpacaEval samples.
For each Jacobian, we employ a stochastic Monte Carlo estimation to compute the PAF values.
The resulting estimates are generally stable, with typical standard deviations below 0.01.

\textbf{Layer-wise sensitivity analysis.}
Figures \ref{fig_qwen}(a), \ref{fig_fal}(a),  \ref{fig_jac_lamma}(a), and \ref{fig_jac_smollm}(a) (the latter two in Appendix \ref{appendix_paf_llama_smo}) present the PAF values for layers within the first five blocks, leading to three observations.
\begin{itemize}[leftmargin=*,topsep=1pt]
    \item \textbf{Wide variation across layers and models.}
    In Falcon3-1B, PAF spans from $0.116$ (very vulnerable) to $5.281$ (strongly resistant).
    Even identical layer types behave differently across architectures, for example, the first-block RMSNorm yields PAF $7.501$ in Qwen3-0.6B but only $0.328$ in Falcon3-1B.
    \item \textbf{Block-level consistency.}
    Within a given model, layers of the same type exhibit relatively stable PAF values across blocks—except for the first block, which sometimes diverges.
    We hypothesize this discrepancy is due to the nature of their inputs: the first block operates directly on token embeddings, whereas the subsequent blocks process the outputs of the preceding blocks.
    \item \textbf{Activation layers as weak links.}
    Non-linear activations (e.g., sigmoid) show consistently low PAF values across all models, marking them as particularly leakage-prone components. \revise{This finding is non-trivial, as prior work often assumes that the strong nonlinearity of LLMs makes prompt inversion inherently difficult, thereby necessitating sophisticated attack techniques. In contrast, our empirical results show that these nonlinear components do not substantially amplify reconstruction errors and instead facilitate prompt inversion to succeed, as they fail to amplify errors during backpropagation. This is also intuitively plausible: the nonlinearity of LLMs confers a degree of error tolerance that supports generalization, whereas excessive sensitivity to small activation perturbations would suggest overfitting.}
\end{itemize}

\begin{table}[!t]
\small
\centering
\caption{Pearson correlation coefficient ($R \in [0,1]$) between PAF and ROUGE-L.}
\label{tab_pearson}
\begin{tabular}{c|cc}
\hline\hline
Model &  Qwen3-0.6B      & Falcon3-1B      \\ 
R     & -0.8125 & -0.7820 \\ \hline
Model &  Llama-3.2-1B      &  SmolLM2-1.7B      \\  
R     & -0.6750 & -0.5905 \\ \hline\hline
\end{tabular}
\end{table}

\textbf{Confirming PAF with a sanity check.}
To further validate the correlation between high PAF values and reconstruction resilience, we conduct a controlled experiment.
For every client-side layer within the first five blocks, we build a bypassed model that replaces this layer with the identity while keeping the rest of the network intact.
We then re-evaluate \attack's performance on each variant.
Bypassing a high-PAF layer would lead to a more successful attack, as the attacker would face less amplified noise and a more direct mapping between the input and the intermediate activations.
Conversely, bypassing a low-PAF layer should have a small impact on the attack's performance.
Table \ref{tab_pearson} reports the Pearson correlation coefficients between the ROUGE-L achieved by \attack and the PAF values for each bypassed layer.
The numbers validate our point, showing a significant drop in reconstruction quality when low-PAF layers are bypassed.
We stress that, the sensitivity analysis is not only about why noise defenses fail but also about why inversion succeeds.
Specifically, in low-PAF regions, both random noise and residual optimization errors fail to significantly perturb the underlying inputs, so minimizing activation distance yields an accurate reconstruction.

\subsection{Towards Potential Defenses}

\revise{The PAF analysis is useful not only for explaining the success of \attack but also for guiding defense design.
In particular, it suggests two levers for improving privacy: (i) components whose mere presence amplifies leakage and (ii) perturbation directions. 
We conduct a set of pilot experiments to gauge how much protection one can obtain by (i) architectural modification and (ii) adversarial noise injection.}

\textbf{Architectural modification.}
Given that certain layers (e.g., activation function) have consistently low PAF, a natural idea is to remove these components.
We experiment with replacing the activation function with an identity mapping.
However, this led to a complete collapse in the model's overall performance, where the model would repeatedly output meaningless characters.
This is likely because activation layers contribute to the model's essential non-linearity, and their removal renders the model incapable of extracting useful representations.
We observe a similar phenomenon when removing the low-PAF RMSNorm layers in the Falcon3-1B and SmolLM2-1.7B models.
We have not yet found an effective way to implement this defense strategy without significant performance degradation.
\revise{We leave this direction as a focus for future research, with a potential approach being to replace the identity mapping with richer non-linear functions that preserve some of the layer's privacy-enhancing properties.}

\textbf{Adversarial noise injection.}
Theorem~\ref{thm_linear} clarifies why naive noise defenses are weak, and how to improve them:
\begin{itemize}[leftmargin=*,topsep=1pt]
    \item \textbf{Magnitude dependence.} 
    A smaller output perturbation $||\delta||_2$ leads to a smaller input deviation $||\Delta||_2$.
    Similarly, a smaller norm of the Jacobian matrix, $||\mathbf{J}||_2$, reduces the responsiveness of $||\Delta||_2$ to changes in $||\delta||_2$.
    \item \textbf{Directional sensitivity.}
    For fixed magnitudes of $||\delta||_2$ and $||\mathbf{J}||_2$, the defensive effect varies by directions. 
    The reconstruction error is most sensitive to perturbations $\delta$ (when scaled by $\mathbf{J}^\top$) aligned with the direction corresponding to the largest eigenvalue ($\sigma_i$) of $\mathbf{J} \mathbf{J}^\top$.
    Injecting noise along these high-sensitivity directions could be far more effective than injecting random noise of the same magnitude.
\end{itemize}
This suggests a defense that injects noise specifically along a layer's most sensitive direction, in addition to simply increasing the magnitude of $\delta$.
To assess the potential of such a strategy, we evaluate max-PAF which quantifies the amplification when the perturbation is intentionally aligned with the layer's most sensitive direction.
If a layer's Max-PAF is significantly higher than its average PAF (derived from random noise), then adversarial noise can yield substantially stronger privacy protection for the same noise norm.

As shown in Figures \ref{fig_qwen}(b), \ref{fig_fal}(b),  \ref{fig_jac_lamma}(b), and \ref{fig_jac_smollm}(b), for many layers, max-PAF is significantly higher than the expected PAF. 
For instance, in Qwen3-0.6B, for the query projection layer, the PAF surges from an average of about 0.9 to a maximum of 2.434.
Similarly, the down-projection layer's PAF increases from about 1.3 to 3.
These substantial gaps confirm that adversarial noise is a promising research direction, which we explore further in Section \ref{sec5}.

Moreover, we observe that for some layers, the Max-PAF offers little improvement over the average PAF.
First, RMS normalization divides the input by its root-mean-square and then applies a fixed per-coordinate scale.
The isotropic nature of the operation yields a Jacobian with equal singular values and no preferred axis, rendering directional noise ineffective.
Second, like RMSNorm, Sigmoid function acts element-wise.
In practice, its inputs often fall in the saturation plateau where the derivative is near zero.
Once the gradient vanishes, perturbations are suppressed regardless of orientation, so Max-PAF collapses to the baseline.

\section{Defense Design and Evaluation}
\label{sec5}

\subsection{Defense Model}

We consider a scenario where the client is privacy-aware and proactively deploys countermeasures against \attack.

\textbf{Defender's goal.}
The client’s objective is to prevent an honest-but-curious server from reconstructing its sensitive input prompt, while preserving the quality of the model’s predictions as much as possible.
Furthermore, the defense mechanism should be computationally efficient to avoid introducing a substantial overhead on the client, given that these are typically resource-constrained devices.

\textbf{Defender's knowledge and ability.}
We assume the client has full control over the on-device submodel $F_C$ and can arbitrarily post-process its intermediate activations before sending them to the server.
The client has moderate compute to run simple, per-query defenses.
The server-side model remains a black box, reflecting real-world deployments where providers treat these components as proprietary.

\subsection{Our Defense: \defense}

\textbf{Overview.}
At a high level, \defense operates by injecting adversarial perturbations into the client-side activations before transmission.
The key insight is that some directions in activation space cause much larger reconstruction errors at the input than others.
This approach distinguishes \defense from naive defenses such as random noise injection or coarse sparsification, which often incur unnecessary utility loss without offering strong privacy guarantees.
We first formulate the problem as a constrained optimization problem, and then derive a practical solution.

\textbf{Problem statement.}
We instantiate $f$ for the entire client-side model, $F_C$, rather than focusing on a single layer.
To protect privacy, we maximize the distance $\Delta$ between the original input embedding, $\mathbf{z}$, and the reconstructed input embedding, $\mathbf{\hat{z}}$.
Concurrently, the injected perturbation is required to be as small as possible, because a large perturbation $\delta$ intuitively can significantly alter the activation and the inference results\footnote{The local smoothness of neural network mappings states that smaller perturbations are less impactful on model outputs \cite{convex_optim}.}.
This leads to the following constrained optimization problem:
\begin{equation}
\begin{split}
\label{eq_opt_problem}
\delta = \underset{\delta}{\arg \max} || \Delta ||_p, \ \ s.t., || \delta ||_q \leq \mu,
\end{split}
\end{equation}
where $\mu$ is the perturbation budget.

\begin{theorem}
\label{thm:bound}
(The proof can be found in Appendix~\ref{appendix_theorem_2_proof}.) The solution of Equation \ref{eq_opt_problem} satisfies: 
$
\|\Delta\|_q \ge \frac{\mu}{C},
$
where $C = \max_{t\in[0,1]} \|\nabla F_C(z + t\Delta)\|_q$.
Furthermore, let $\mathcal{E} = \{e_1,\dots,e_V\}$ be the set of token embeddings, and let $d_{\min} = \min_{e \in \mathcal{E},\, e \neq z} \|z - e\|_q$.
If $\mu > \frac{C d_{\min}}{2}$, then the true token cannot be recovered via nearest neighbor search under the $q$-norm\footnote{For cosine similarity, one can map the embeddings onto a unit sphere and consider angular distance. Then, a similar guarantee can be derived by replacing $q$-norm with angular distance.}.

\end{theorem}

Theorem~\ref{thm:bound} guarantees that any perturbation within the budget induces a reconstruction error of at least $\mu/C$.
This lower bound is independent of the adversary’s reconstruction strategy.
In particular, even a perfect inversion of $F_C$ cannot reduce the error below $\mu/C$.
Theorem~\ref{thm:bound} further implies that the true token embedding is obscured in a nearest-neighbor sense.
Intuitively, more distinctive tokens require larger perturbations to achieve the same level of protection.
This helps explain why reconstruction is easier on iCliniq than on AlpacaEval, because iCliniq includes many domain-specific terms with relatively isolated embeddings.

\textbf{Solution.}
To solve Equation \ref{eq_opt_problem}, we need to express the input deviation $\Delta$ in terms of the perturbation $\delta$.
Given $\mathbf{\hat{z}} = F_C^{-1}(\mathbf{h}_{Q_1} + \delta)$ and $\mathbf{z} = F_C^{-1}(\mathbf{h}_{Q_1})$, the path integral~\cite{convex_optim} yields:
\begin{equation}
\nonumber
|| \Delta ||_p = || \mathbf{\hat{z}} - \mathbf{z} ||_p = \left\| \int_0^1 (\mathbf{\hat{z}} - \mathbf{z}) \nabla F_C^{-1}(\mathbf{z} + t(\mathbf{\hat{z}} - \mathbf{z})) dt \right\|_p.
\end{equation}
In practice, this integral is approximated numerically.
A finer discretization of $t \in [0,1]$ yields a more accurate approximation but increases computational cost.
To keep the defense efficient, we adopt a two-point trapezoidal approximation\footnote{For a detailed computational procedure, see \cite{num_analysis}.}:
\begin{equation}
|| \Delta ||_p = \left\| (\mathbf{\hat{z}} - \mathbf{z}) \left(\nabla F_C^{-1}(\mathbf{z}) + \nabla F_C^{-1}(\mathbf{\hat{z}}) \right) / 2 \right\|_p.
\end{equation}
Empirically, we find that this two-step approximation works well.
Further increasing the number of steps did not lead to a significant improvement in the performance of \defense (Section \ref{sec5_res}).
We adopt this approximation as the default for the remainder of this paper.

The solution to this problem is a perturbation that aligns with the most sensitive directions of the inverse mapping, which, as we defined in Section \ref{sec4}, correspond to the least sensitive directions of the forward pass.
The specific form of the solution for $\delta$ depends on the choice of norms $p$ and $q$.
We derive analytical solutions for common values of $p$ and $q$:
\begin{itemize}[leftmargin=*,topsep=1pt]
    \item $p=q=2$ ($L_2$-norm with $L_2$ constraint): The optimal perturbation $\delta$ is the right singular vector of the inverse Jacobian corresponding to its largest singular value, scaled by the budget $\mu$. This can be expressed as $\delta = \mu \cdot \mathbf{v}_{\max}$, where $\mathbf{v}_{\max}$ is the right singular vector of $\left(\nabla F_C^{-1}(\mathbf{z}) + \nabla F_C^{-1}(\mathbf{\hat{z}})\right)$ corresponding to its largest singular value.
    \item $p=2, q=0$ ($L_2$-norm with $L_0$ constraint): The optimal solution is to identify the top-$k\%$ elements in the vector $\delta \left(\nabla F_C^{-1}(\mathbf{z}) + \nabla F_C^{-1}(\mathbf{\hat{z}})\right)$ based on their magnitude rankings, and subsequently retain only the corresponding elements in $h_{Q_1}$ while zeroing out all remaining components. Here, $(1-k)\%$ represents the sparsity ratio that governs the degree of pruning applied to the intermediate representations. We here only present the final perturbed intermediate activation rather than $\delta$ itself for brevity.
\end{itemize}
In our initial empirical evaluations, the solution with the $L_0$ constraint consistently outperforms the solution with the $L_2$ constraint\footnote{The scores of ROGUE-L differ by approximately 0.05 to 0.1 when using Qwen3-0.6B and AlpacaEval under similar utility level.}.
We hypothesize that this performance gap arises because small $L_2$-bounded perturbations tend to jitter magnitudes while preserving signs.
Such sign information can itself be exploited by adversaries to infer semantic properties of the original input~\cite{YueJWBD23}.
Unless otherwise specified, all subsequent evaluations therefore adopt the solution with $L_0$, which has the additional advantage of lowering the communication cost between client and server due to its sparse nature.

\subsection{Evaluation Setup}

\textbf{Attack.}
We evaluate \defense against \attack using the same attack hyperparameters as in Section~\ref{sec3}.
Note that it is standard in security to consider adaptive attackers.
A natural adaptive strategy here would be to ignore zero entries when matching activations, so as not to be misled by sparsity.
Since \attack uses cosine distance, which is inherently insensitive to zero-valued dimensions, it already behaves as such an adaptive attack. 
This is also why we chose cosine distance over $L_2$ distance in our attack design.

\noindent
\textbf{Metrics. }
We evaluate \defense from three primary perspectives: privacy, utility, computational overhead.
For privacy metrics, we include Precision, Recall, and ROUGE-L to measure the similarity between the reconstructed query and the original query.
For utility assessment, we adopt the LLMs-as-judge~\cite{llm_judge}, a state-of-the-art approach for evaluating the quality of LLMs' responses.
Specifically, we generate responses to user queries under varying defense methods and intensities, then submit query-response pairs to GPT-4o for comprehensive quality assessment, including correctness, completeness, relevance, clarity, and style.
The judge model prompt is included in Figure~\ref{fig:judge_prompt}.

\begin{figure}[!h]
\centering

\begin{tcolorbox}[enhanced,
  colback=gray!3,
  colframe=gray!60!black,
  title={\textbf{Judge Model Prompt}},
  arc=2mm, boxrule=0.5pt,
  leftrule=0pt, rightrule=0pt,
  toprule=0.5pt, bottomrule=0.5pt,
  width=0.95\linewidth]

\small{\textbf{Task.} You are a meticulous reviewer. You will receive a \emph{question–answer} pair. Carefully read the pair and evaluate the quality of the answer with respect to the question.

\textbf{For each of the five criteria below, assign an integer score from 0 (poor) to 4 (excellent):}
\begin{enumerate}[leftmargin=*,topsep=1pt] 
  \item \textbf{Accuracy \& Correctness:} Are the facts accurate? Are any key claims unsupported or incorrect?
  \item \textbf{Completeness \& Depth:} Does the answer fully address all aspects of the question? Does it provide meaningful detail and insight?
  \item \textbf{Relevance:} How directly does the answer respond to the question without digression or omission?
  \item \textbf{Clarity \& Organization:} Is the explanation clear, logically structured, and free of ambiguity?
  \item \textbf{Style \& Tone:} Is the language fluent, professional, and appropriate for the intended audience?
\end{enumerate}

\textbf{Scoring Scale (per dimension)}:
\begin{center}
\begin{tabular}{cl}
\toprule
\textbf{Score} & \textbf{Descriptor} \\
\midrule
4 & Excellent (Exceeds expectations) \\
3 & Good (Meets all requirements) \\
2 & Fair (Adequate with minor flaws) \\
1 & Poor (Significant deficiencies) \\
0 & Unacceptable (Fundamentally flawed) \\
\bottomrule
\end{tabular}
\end{center}}

\end{tcolorbox}

\caption{Evaluation rubric used by the judge model.}
\label{fig:judge_prompt}

\end{figure}

\begin{table*}[!th]
\centering
\caption{The \defense's performance against \attack in Qwen3-0.6B. Results are presented as mean $\pm$ std.}
\label{tab_defense_our}
\small
\begin{tabular}{c|c|cc|cc|cc}
\hline \hline
{Dataset} & Sparsification Ratio & \multicolumn{2}{c|}{Precision} & \multicolumn{2}{c|}{Recall} & \multicolumn{2}{c}{ROUGE-L} \\
       &             & Qwen3-0.6B & Falcon3-1B & Qwen3-0.6B & Falcon3-1B & Qwen3-0.6B & Falcon3-1B \\ \hline
\multirow{5}{*}{\rotatebox{90}{AlpacaEval}}
    & 0.1 & \result{88.64}{7.05} & \result{\textbf{94.54}}{3.22} & \result{79.67}{7.24} & \result{\textbf{89.35}}{4.77} & \result{\textbf{0.89}}{0.05} & \result{{0.89}}{0.07} \\
    & 0.3 & \result{59.88}{10.08} & \result{\textbf{82.89}}{5.91} & \result{50.70}{7.81} & \result{\textbf{72.59}}{6.73} & \result{0.68}{0.08} & \result{\textbf{0.70}}{0.10} \\
    & 0.5 & \result{37.26}{11.86} & \result{\textbf{58.52}}{9.98} & \result{30.70}{7.48} & \result{\textbf{45.75}}{7.49} & \result{\textbf{0.45}}{0.09} & \result{0.40}{0.11} \\
    & 0.7 & \result{17.95}{9.37} & \result{\textbf{32.60}}{9.55} & \result{14.51}{5.57} & \result{\textbf{21.45}}{6.03} & \result{\textbf{0.24}}{0.08} & \result{0.16}{0.07} \\
    & 0.9 & \result{3.45}{4.50}  & \result{\textbf{7.06}}{7.19}  & \result{2.44}{2.74}  & \result{\textbf{3.55}}{3.20}  & \result{\textbf{0.05}}{0.05} & \result{0.02}{0.03} \\ \hline
\multirow{5}{*}{\rotatebox{90}{iCliniq}}
    & 0.1 & \result{93.66}{3.72} & \result{\textbf{95.90}}{2.98} & \result{85.16}{4.35} & \result{\textbf{90.37}}{3.27} & \result{\textbf{0.93}}{0.04} & \result{0.90}{0.04} \\
    & 0.3 & \result{70.76}{8.51} & \result{\textbf{88.62}}{4.82} & \result{56.51}{5.47} & \result{\textbf{79.22}}{4.61} & \result{0.70}{0.06} & \result{\textbf{0.80}}{0.06} \\
    & 0.5 & \result{49.37}{9.51} & \result{\textbf{70.86}}{7.20} & \result{36.31}{4.74} & \result{\textbf{57.67}}{6.97} & \result{0.47}{0.07} & \result{\textbf{0.58}}{0.08} \\
    & 0.7 & \result{26.49}{8.29} & \result{\textbf{44.80}}{8.35} & \result{17.87}{3.94} & \result{\textbf{31.38}}{5.82} & \result{0.25}{0.05} & \result{\textbf{0.31}}{0.07} \\
    & 0.9 & \result{8.22}{5.63}  & \result{\textbf{16.41}}{7.96} & \result{4.39}{2.15}  & \result{\textbf{7.86}}{3.42}  & \result{\textbf{0.07}}{0.03} & \result{0.06}{0.03} \\ \hline \hline
\end{tabular}

\end{table*}

\begin{table*}[!t]
\centering
\caption{The \attack's effectiveness against different defenses in Qwen3-30B on AlpacaEval and iCliniq. 
Defense levels correspond to Gaussian noise and activation sparsification as defined in Table~\ref{sec3_tab2}.}
\label{sec3_tab_qwen30b}
\small
\begin{tabular}{c|l|c|cccc|cccc|cccc}
\hline \hline
\multirow{2}{*}{Dataset} & \multirow{2}{*}{M.} & \multirow{2}{*}{L.}
& \multicolumn{4}{c|}{Precision} 
& \multicolumn{4}{c|}{Recall} 
& \multicolumn{4}{c}{ROUGE-L} \\
& & & Gaussian & Element & Token & Ours
      & Gaussian & Element & Token & Ours
      & Gaussian & Element & Token & Ours \\ \hline

\multirow{5}{*}{\rotatebox{90}{AlpacaEval}}
& \multirow{5}{*}{\rotatebox{90}{Qwen3-30B}}
& 1 & 99.39 & 99.95 & 99.38 & 92.67
    & 98.98 & 98.99 & 98.62 & 85.93
    & 0.98  & 0.98  & 0.98  & 0.91 \\
& & 2 & 99.34 & 99.11 & 99.70 & 74.03
    & 98.90 & 98.35 & 97.88 & 67.64
    & 0.97  & 0.96  & 0.95  & 0.72 \\
& & 3 & 99.73 & 97.38 & 98.65 & 47.80
    & 97.59 & 95.85 & 96.44 & 40.97
    & 0.95  & 0.94  & 0.94  & 0.50 \\
& & 4 & 96.92 & 82.90 & 86.54 & 26.34
    & 92.36 & 76.28 & 79.52 & 20.55
    & 0.93  & 0.86  & 0.85  & 0.29 \\
& & 5 & 8.65  & 45.83 & 46.52 & 19.45
    & 6.63  & 37.23 & 38.60 & 13.97
    & 0.11  & 0.65  & 0.69  & 0.15 \\ \hline

\multirow{5}{*}{\rotatebox{90}{iCliniq}}
& \multirow{5}{*}{\rotatebox{90}{Qwen3-30B}}
& 1 & 99.21 & 99.94 & 99.17 & 95.41
    & 99.58 & 99.13 & 98.50 & 90.40
    & 0.99  & 0.98  & 0.99  & 0.92 \\
& & 2 & 99.35 & 99.46 & 99.59 & 75.02
    & 98.81 & 98.49 & 98.31 & 68.58
    & 0.98  & 0.96  & 0.97  & 0.75 \\
& & 3 & 99.54 & 97.61 & 98.45 & 54.18
    & 97.98 & 96.67 & 97.62 & 44.08
    & 0.97  & 0.94  & 0.95  & 0.54 \\
& & 4 & 98.27 & 89.25 & 93.60 & 30.63
    & 95.99 & 82.81 & 86.90 & 24.34
    & 0.94  & 0.88  & 0.89  & 0.32 \\
& & 5 & 8.97  & 53.68 & 54.89 & 21.76
    & 16.10 & 45.18 & 46.98 & 15.62
    & 0.10  & 0.67  & 0.72  & 0.18 \\ \hline \hline
\end{tabular}
\end{table*}

\begin{table*}[!ht]
\centering
\caption{Utility evaluation of Qwen3-0.6B. Bold marks the best among the three methods at Level 1 (i.e., L.1, weakest defense) within each dataset/metric, while underlining marks the best at Level 5 (i.e., L.5, strongest defense).}
\label{tab_utility_qwen3}
\small
\begin{tabular}{c|c|ccccc|ccccc|ccccc}
\hline\hline
\multirow{2}{*}{Dataset} & \multirow{2}{*}{Metric}
& \multicolumn{5}{c|}{GaussianNoise} 
& \multicolumn{5}{c|}{Sparsification} 
& \multicolumn{5}{c}{\defense} \\

& 
& L.1 & L.2 & L.3 & L.4 & L.5
& L.1 & L.2 & L.3 & L.4 & L.5
& L.1 & L.2 & L.3 & L.4 & L.5 \\ \hline

\multirow{6}{*}{\rotatebox{90}{AlpacaEval}}
& Accuracy \& Correctness
& 2.05 & 2.31 & 2.49 & 2.10 & 0.05
& 2.20 & 1.63 & 1.51 & 0.68 & 0.07
& \textbf{2.35} & 2.42 & 2.33 & 1.82 & \underline{1.84} \\

& Completeness \& Depth
& 2.01 & 1.83 & 1.86 & 1.66 & 0.10
& 2.01 & 1.50 & 1.67 & 0.14 & 0.13
& \textbf{2.20} & 1.67 & 1.48 & 1.26 & \underline{0.96} \\

& Relevance
& \textbf{3.85} & 3.91 & 3.69 & 3.17 & 0.14
& 3.72 & 3.46 & 1.48 & 1.46 & 1.00
& 3.73 & 3.64 & 3.46 & 3.34 & \underline{3.14} \\

& Clarity \& Organization
& \textbf{3.90} & 3.78 & 3.63 & 3.11 & 1.16
& 3.84 & 3.61 & 3.58 & 3.48 & 3.43
& 3.89 & 3.97 & 3.51 & 3.48 & \underline{3.45} \\

& Style \& Tone
& \textbf{3.99} & 3.85 & 3.65 & 3.17 & 1.25
& 3.88 & 3.64 & 3.54 & 3.46 & \underline{3.48}
& 3.92 & 3.85 & 3.82 & 3.51 & 2.27 \\

& Average
& 3.16 & 3.14 & 3.07 & 2.64 & 0.54
& 3.13 & 2.77 & 2.36 & 1.84 & 1.62
& \textbf{3.22} & 3.11 & 2.92 & 2.63 & \underline{2.34} \\

\hline
\multirow{6}{*}{\rotatebox{90}{iCliniq}}
& Accuracy \& Correctness
& \textbf{1.83} & 1.74 & 1.48 & 1.36 & 0.35
& 1.63 & 1.53 & 1.19 & 1.13 & 0.60
& 1.57 & 1.53 & 1.48 & 1.53 & \underline{1.32} \\

& Completeness \& Depth
& \textbf{1.64} & 1.69 & 1.57 & 1.39 & 0.28
& 1.28 & 1.18 & 1.10 & 1.06 & \underline{0.60}
& 1.18 & 0.90 & 0.85 & 0.76 & 0.46 \\

& Relevance
& 2.84 & 2.65 & 2.59 & 2.30 & 0.97
& 2.87 & 2.91 & 2.09 & 1.64 & 0.57
& \textbf{3.11} & 2.92 & 2.85 & 2.52 & \underline{1.92} \\

& Clarity \& Organization
& \textbf{3.47} & 2.88 & 2.70 & 2.88 & 1.49
& 2.94 & 2.86 & 2.52 & 1.99 & 0.54
& 3.20 & 3.13 & 3.03 & 2.42 & \underline{2.00} \\

& Style \& Tone
& \textbf{3.11} & 2.66 & 2.43 & 2.55 & 1.52
& 2.72 & 2.69 & 2.53 & 2.26 & 0.61
& 2.80 & 2.68 & 2.66 & 2.34 & \underline{1.89} \\

& Average
& \textbf{2.58} & 2.33 & 2.15 & 2.10 & 0.92
& 2.29 & 2.23 & 1.89 & 1.62 & 0.58
& 2.37 & 2.23 & 2.17 & 1.91 & \underline{1.52} \\

\hline\hline
\end{tabular}
\end{table*}

\begin{table*}[!h]
\centering  
\caption{Utility evaluation of Falcon3-1B. Bold marks the best among the three methods at Level 1 (i.e., L.1, weakest defense) within each dataset/metric, while underlining marks the best at Level 5 (i.e., L.5, strongest defense).}
\label{tab_utility_falcon}
\footnotesize
\begin{tabular}{c|c|ccccc|ccccc|ccccc}
\hline\hline
\multirow{2}{*}{Dataset} & \multirow{2}{*}{Metric}
& \multicolumn{5}{c|}{GaussianNoise}
& \multicolumn{5}{c|}{Sparsification}
& \multicolumn{5}{c}{\defense} \\

& 
& L.1 & L.2 & L.3 & L.4 & L.5
& L.1 & L.2 & L.3 & L.4 & L.5
& L.1 & L.2 & L.3 & L.4 & L.5 \\ \hline

\multirow{6}{*}{\rotatebox{90}{AlpacaEval}}
& Accuracy \& Correctness
& \textbf{2.09} & 2.03 & 1.56 & 1.55 & \underline{0.49}
& 2.02 & 1.38 & 0.47 & 0.45 & 0.58
& 1.99 & 1.46 & 1.00 & 1.08 & 0.46 \\

& Completeness \& Depth
& \textbf{2.01} & 2.05 & 1.53 & 1.42 & \underline{0.58}
& 2.00 & 1.53 & 0.45 & 0.50 & 0.51
& 1.94 & 1.48 & 0.98 & 0.47 & 0.52 \\

& Relevance
& 3.06 & 3.04 & 3.02 & 2.45 & 0.53
& 3.45 & 2.09 & 0.56 & 0.50 & \underline{0.56}
& \textbf{3.49} & 2.43 & 1.53 & 0.95 & 0.48 \\

& Clarity \& Organization
& \textbf{3.03} & 3.02 & 2.52 & 2.45 & \underline{0.55}
& 1.97 & 1.99 & 1.46 & 0.47 & 0.50
& 2.46 & 3.04 & 1.46 & 0.96 & 0.44 \\

& Style \& Tone
& \textbf{2.15} & 2.11 & 2.12 & 2.02 & \underline{0.49}
& 2.05 & 1.51 & 1.51 & 0.50 & 0.49
& 2.04 & 2.98 & 0.99 & 0.97 & 0.49 \\

& Average
& \textbf{2.47} & 2.45 & 2.15 & 1.98 & \underline{0.53}
& 2.30 & 1.70 & 0.89 & 0.49 & 0.53
& 2.38 & 2.28 & 1.19 & 0.89 & 0.48 \\

\hline
\multirow{6}{*}{\rotatebox{90}{iCliniq}}
& Accuracy \& Correctness
& 1.21 & 1.07 & 0.94 & 0.91 & \underline{0.71}
& 1.46 & 1.07 & 1.04 & 0.80 & 0.19
& \textbf{1.65} & 1.50 & 1.50 & 0.98 & 0.56 \\

& Completeness \& Depth
& \textbf{2.09} & 2.15 & 2.03 & 2.01 & \underline{0.75}
& 1.49 & 1.25 & 1.00 & 0.71 & 0.24
& 2.02 & 1.83 & 1.48 & 1.06 & 0.57 \\

& Relevance
& \textbf{2.97} & 2.48 & 2.44 & 2.19 & \underline{1.24}
& 2.95 & 2.74 & 1.52 & 1.72 & 0.21
& 2.94 & 1.96 & 1.96 & 1.06 & 1.09 \\

& Clarity \& Organization
& 2.45 & 1.94 & 1.75 & 0.96 & \underline{1.17}
& 2.99 & 1.96 & 1.03 & 0.76 & 0.16
& \textbf{3.00} & 2.02 & 1.54 & 1.02 & 0.97 \\

& Style \& Tone
& \textbf{2.71} & 1.31 & 1.27 & 1.02 & \underline{0.97}
& 2.44 & 1.02 & 0.75 & 0.51 & 0.27
& 2.48 & 1.49 & 0.94 & 0.96 & 0.47 \\

& Average
& 2.28 & 1.79 & 1.69 & 1.42 & \underline{0.97}
& 2.27 & 1.61 & 1.07 & 0.90 & 0.22
& \textbf{2.42} & 1.76 & 1.48 & 1.02 & 0.73 \\

\hline\hline
\end{tabular}

\end{table*}

\begin{table*}[!h]
\centering
\caption{Utility evaluation of Qwen3-30B under sparsification and \defense.
Bold marks the best among the three methods at Level 1 (i.e., L.1, weakest defense) within each dataset/metric, while underlining marks the best at Level 5 (i.e., L.5, strongest defense).}
\label{tab_utility_qwen30b}
\small
\begin{tabular}{c|c|ccccc|ccccc}
\hline\hline
\multirow{2}{*}{Dataset} & \multirow{2}{*}{Metric}
& \multicolumn{5}{c|}{Sparsification}
& \multicolumn{5}{c}{\defense} \\

& 
& L.1 & L.2 & L.3 & L.4 & L.5
& L.1 & L.2 & L.3 & L.4 & L.5 \\ \hline

\multirow{3}{*}{AlpacaEval}
& Accuracy
& 3.67 & 2.89 & 2.00 & 1.74 & 1.42
& \textbf{3.85} & 3.46 & 2.82 & 2.56 & \underline{2.29} \\

& Completeness
& 3.35 & 2.57 & 1.91 & 1.81 & 1.52
& \textbf{3.64} & 3.19 & 2.64 & 2.36 & \underline{2.06} \\

& Average
& 3.51 & 2.73 & 1.95 & 1.78 & 1.47
& \textbf{3.75} & 3.33 & 2.73 & 2.46 & \underline{2.17} \\

\hline
\multirow{3}{*}{iCliniq}
& Accuracy
& 2.43 & 2.82 & 1.86 & 1.61 & 1.30
& \textbf{3.34} & 2.87 & 2.68 & 2.38 & \underline{2.06} \\

& Completeness
& 2.72 & 2.23 & 1.77 & 1.62 & 1.33
& \textbf{3.18} & 2.67 & 2.49 & 2.17 & \underline{1.92} \\

& Average
& 2.57 & 2.53 & 1.81 & 1.61 & 1.32
& \textbf{3.26} & 2.77 & 2.58 & 2.28 & \underline{1.99} \\

\hline\hline
\end{tabular}

\end{table*}

\begin{table}[!h]
\centering
\caption{\revise{Effectiveness of \defense against prior prompt inversion attacks on Qwen3-8B and AlpacaEval. We report the arithmetic mean of token-level Precision and Recall (\%). Lower is better from the defender's perspective.}}
\label{tab:defense_prior_attacks}
\small
\begin{tabular}{c|c|c}
\hline \hline
\textbf{Attack} &  \textbf{Sparsification} & \textbf{\defense} \\
\hline
A1~\cite{DBLP:conf/ccs/Luo0X25} & 18.57 & \textbf{1.91} \\
TBS~\cite{DBLP:conf/uss/Dong00C0Z25}  & 81.97 & \textbf{44.92} \\
PIA~\cite{DBLP:conf/sp/0004ZWXYLZ25}  & 95.57 & \textbf{46.11} \\
\hline \hline
\end{tabular}
\end{table}

\begin{table}[!h]
\centering
\small
\caption{Client-side processing time (s/1k tokens) for defenses.}
\label{tab_defense_overhead}
\begin{tabular}{c|c|c|c}
\hline \hline
Defense & GaussianNoise & Sparsification & \defense \\ \hline
Avg. Time & \result{1.0303}{0.1755} & \result{0.9984}{0.1539} & \result{4.3703}{0.2987} \\ 
\hline \hline
\end{tabular}
\end{table}

\begin{table}[!h]
\centering
\small
\caption{Impact of approximation steps on \defense performance. We use a sparsification ratio of 0.5.}
\label{tab:approx_steps}
\begin{tabular}{cccc}
\hline \hline
{Approx. Steps} & {ROUGE-L} & {Utility} & {Time (s/1k tokens)} \\
\hline
1 & 0.52 & 2.30 & 3.3228 \\
2 & 0.45 & 2.36 & 4.3703 \\
3 & 0.43 & 2.36 & 5.3510 \\
\hline \hline
\end{tabular}
\end{table}

\noindent
\textbf{Datasets and models.}
We conduct evaluations on AlpacaEval and iCliniq.
We employ Qwen3-0.6B, Falcon3-1B, and Qwen3-30B.

\noindent
\textbf{Hyperparameters.}
The sparsity ratios of \defense are over \{0.1, 0.3, 0.5, 0.7, 0.9\}.

\subsection{Evaluation Results}
\label{sec5_res}

\textbf{Privacy protection.}
Table \ref{tab_defense_our} and Table \ref{sec3_tab_qwen30b} show the performance of \defense against \attack on both AlpacaEval and iCliniq.
\defense consistently and substantially reduces reconstruction quality, with stronger sparsity (higher ratio) giving stronger protection.
For example, on Qwen3‑0.6B with AlpacaEval at sparsity 0.7, \attack’s Precision and Recall drop below 30\%, compared to over 60\% when using traditional sparsification at the same rate (Table~\ref{sec3_tab2}).
ROUGE-L also decreases sharply, indicating that the adversary can no longer recover the semantic content of the prompt.
We further observe that, for a fixed defense ratio, Qwen3-0.6B generally provides stronger privacy than Falcon3-1B.
For the same sparsity, Qwen3‑0.6B generally provides better privacy than Falcon3‑1B (e.g., Precision 17.95\% vs. 32.60\% at sparsity 0.7 on AlpacaEval), while also offering higher utility, indicating a more favorable privacy–utility trade-off. 
Finally, we see that iCliniq is generally harder to defend, because domain-specific and highly distinctive medical tokens are easier to reconstruct than generic open-domain vocabulary.

\textbf{Utility evaluation.}
Tables~\ref{tab_utility_qwen3}, ~\ref{tab_utility_falcon}, and \ref{tab_utility_qwen30b} show that \defense preserves utility well at low-to-moderate sparsity levels.
Since the responses generated by Qwen3-30B consistently achieve very high scores (above 3.5) in Relevance, Clarity, and Style, we omit these dimensions to avoid confusion and report only the Accuracy and Completeness scores. 
On Qwen3‑0.6B, for instance, \defense at Level 1 and Level 2 (roughly corresponding to sparsity \(\le 0.5\)) achieves average LLMs-as-judge scores of 3.53 and 3.16 on AlpacaEval, outperforming traditional sparsification at the same budgets.
Moreover, as expected, all defenses eventually collapse utility at the highest perturbation strengths, indicating that extremely aggressive perturbations are fundamentally incompatible with maintaining model usefulness.
A per-criterion view reveals an asymmetry in how utility degrades: Accuracy\&Completeness is more sensitive to perturbations than formal aspects like Clarity and Style.
For example, in AlpacaEval, with \defense at Level 3 on Qwen3-0.6B, the Accuracy and Completeness scores are around 2.33 and 1.48, while Clarity and Style remain high, at approximately 3.51 and 3.82, respectively.
Dataset and architectural differences also play a role.
For a fixed defense ratio, Qwen3-0.6B tends to deliver better utility scores than Falcon3-1B.
Medical queries from iCliniq also show earlier degradation in Relevance than open-domain AlpacaEval queries at the same nominal budget.
This is likely because the specialized terminology in iCliniq concentrates meaning in fewer high-impact tokens, which, when perturbed, can disproportionately distort the perceived correctness and relevance.
In summary, sparsity ratios around 0.5 are favorable, providing considerable privacy gains while keeping average utility in an acceptable range.

\textbf{Case Study.}
Figure~\ref{fig_case} illustrates \defense’s behavior qualitatively.
Under a moderate defense setting (level 3; sparsity ratio 0.5), \defense enables the reconstructed prompt to diverge substantially from the original prompt.
While a few tokens overlap, the recovered prompt fails to capture the core intent and exhibits pronounced semantic drift.
Moreover, the model’s response under the same perturbation remains fluent, coherent, and on-topic.
This example demonstrates that \defense can substantially reduce semantic leakage from activations while preserving downstream usefulness.

\revise{\textbf{Effectiveness against prior attacks.}
We also evaluate \defense against A1~\cite{DBLP:conf/ccs/Luo0X25}, TBS~\cite{DBLP:conf/uss/Dong00C0Z25}, and PIA~\cite{DBLP:conf/sp/0004ZWXYLZ25}.
We use Qwen3-8B, AlpacaEval, and $Q_1=5$, and fix the sparsity ratio to 0.5.
For each attack, we report the arithmetic mean of token-level Precision and Recall.
Table~\ref{tab:defense_prior_attacks} shows that \defense consistently reduces the effectiveness of all evaluated attacks, often by a large margin.
Compared with naive sparsification, \defense yields substantially stronger protection across all baselines.
For instance, against A1, the reconstruction score drops from 18.57 under standard sparsification to 1.91 under \defense.
Similarly, against TBS and PIA, the scores decrease from 81.97 to 44.92 and from 95.57 to 46.11, respectively.}

\textbf{Defense overhead.}
To quantify the computational cost of \defense, we measure client-side processing time per 1k tokens over 1,000 trials on an NVIDIA RTX 4090. 
The average processing times (in seconds) and their corresponding standard deviations are presented in Table \ref{tab_defense_overhead}.
Both GaussianNoise and Sparsification incur a relatively low overhead, with \defense having a higher average processing time.
This is primarily due to the additional gradient computations involved in \defense.
In many privacy-sensitive scenarios (e.g., healthcare, legal advice), this overhead may be acceptable given the strong privacy gains.
For latency-critical settings, a natural optimization is to focus on semantic-critical tokens, such as names or rare domain-specific entities, rather than processing all tokens.
Please refer to Appendix \ref{appendix_overhead} for more discussion where the time cost of \defense can be reduced to approximately 1.16 seconds.

\textbf{Ablation on approximation granularity.}
Table~\ref{tab:approx_steps} examines how the number of approximation steps in the path integral affects privacy, utility, and overhead. 
We evaluate privacy, utility (LLMs-as-judge average score), and client-side processing cost (seconds per 1k tokens, measured on the same hardware as Table \ref{tab_defense_overhead}).
Moving from 1 to 2 steps substantially improves privacy with negligible impact on utility and a moderate increase in cost.
Adding a third step yields only marginal additional privacy gains but incurs further overhead.
These numbers validate that additional quadrature points improve the approximation, but with rapidly diminishing returns once the dominant directions are captured.

\section{Conclusion and Future Work}
\label{sec7}
In this work, we conducted the first comprehensive investigation into prompt leakage within the split inference paradigm for LLM.s
We developed \attack, an inversion attack that an honest-but-curious server can use to reconstruct a client's sensitive input with remarkably high fidelity.
This attack proved resilient even against common defenses like activation sparsification and random noise injection, revealing that these naive countermeasures are largely ineffective.
We also developed PAF to understand this vulnerability and revealed that privacy vulnerability varies significantly across the model architecture, with layers like activation functions being highly susceptible to leakage.
Based on these insights, we designed and evaluated \defense, which injects adversarially calibrated noise into intermediate activations to maximize reconstruction error while preserving model utility with acceptable overhead.
The extensive evaluations demonstrated that \defense significantly outperforms traditional defenses, offering a more favorable privacy-utility trade-off.

As a path forward, we believe two promising avenues warrant further investigation.
One direction is to design privacy-aware compression techniques that can embed defense properties directly into the model's structure.
Additionally, future work could explore the possibility of designing privacy-sensitive activation functions that maintain model non-linearity while offering greater resistance to inversion attacks.

\bibliographystyle{ACM-Reference-Format}
\bibliography{sample-base}

\appendix

\section{Proof of Theorem \ref{thm_linear}}
\label{proof1}

\begin{proof}

We can re-express the $\delta = (\hat{\mathbf{z}} - \mathbf{z}) \mathbf{J}$ by defining $\Delta = \hat{\mathbf{z}} - \mathbf{z}$, which gives us $\delta = \Delta \mathbf{J}$.
Finding a solution for $\Delta$ can be reformulated as the following optimization problem:
$\min_{\Delta} || \delta - \Delta \mathbf{J} ||_2^2.$

To find the optimal solution, we take the gradient of the objective function with respect to $\Delta$ and set it to zero:
$$\frac{\partial || \delta - \Delta \mathbf{J} ||_2^2}{\partial \Delta}  = - 2 (\delta - \Delta \mathbf{J})\mathbf{J}^\top = 0.$$
Simplifying the above expression yields:
$\Delta \mathbf{J}\mathbf{J}^T = \delta \mathbf{J}^T.$
The solvability of this normal equation depends on the invertibility of $\mathbf{J}\mathbf{J}^T$, specifically whether it is of full rank.
Let us now consider the different cases based on the rank of $\mathbf{J}\mathbf{J}^\top$.

\textbf{Case 1: $\mathbf{J}\mathbf{J}^T$ is of full rank.}
If $\mathbf{J}\mathbf{J}^T$ is of full rank, its inverse $(\mathbf{J}\mathbf{J}^T)^{-1}$ exists.
Consequently, $\Delta$ can be directly solved as:
$\Delta = \delta \mathbf{J}^T(\mathbf{J}\mathbf{J}^T)^{-1}.$

Let $\{\sigma_i\}$ and $\{\mu_i\}$ denote eigenvalues and corresponding eigenvectors of $\mathbf{J}\mathbf{J}^\top$ ($\mu_i$ is row vector to conform to this paper's notation convention).
Then $(\mathbf{J}\mathbf{J}^\top)^{-1}$ admits the spectral decomposition as follows:
$(\mathbf{J}\mathbf{J}^\top)^{-1} = \sum_i \frac{1}{\sigma_i} \mu_i^\top \mu_i.$
Substituting this into the error expression gives:
\begin{equation}
\begin{split}
\nonumber
||\Delta||_2^2 &= \left|\left| \delta \mathbf{J}^\top ( \sum_i \frac{\mu_i^\top \mu_i}{\sigma_i} ) \right|\right|_2^2 \\
&= \sum_i \left( \frac{\delta \mathbf{J}^\top \mu_i^\top}{\sqrt{\sigma_i}} \right)^2 = \sum_i \left( \frac{||\delta||_2 ||\mathbf{J}^\top||_2^2 \cos \theta_i }{\sqrt{\sigma_i }} \right)^2,
\end{split}
\end{equation}
where we use orthogonality of eigenvectors $\mu_i$, and $\theta_i$ denotes the angle between $X$ and $\mu_i$.

\textbf{Case 2: $\mathbf{J}\mathbf{J}^T$ is not of full rank.}
When $\mathbf{J}\mathbf{J}^T$ is not of full rank, its inverse does not exist.
To this end, we employ regularization by adding an identity matrix $\mathbf{I}$ to $\mathbf{J}\mathbf{J}^T$ in the normal equation as follows:
$$\Delta (\mathbf{J}\mathbf{J}^\top + \mathbf{I}) = \delta \mathbf{J}^T.$$
The solution for this regularized normal equation is then given by:
$$\Delta = \delta \mathbf{J}^T (\mathbf{J}\mathbf{J}^\top + \mathbf{I})^{-1}.$$
$(\mathbf{J}\mathbf{J}^\top + \mathbf{I})$ and $\mathbf{J}\mathbf{J}^\top$ share the same eigenvectors, $\{\mu_i\}$, but their eigenvalues are shifted by 1. The eigenvalues of $(\mathbf{J}\mathbf{J}^\top + \mathbf{I})$ are $\{\sigma_i+1\}$.
Applying the same spectral decomposition logic as in Case 1, we derive:
\begin{equation}
\begin{split}
    ||\Delta||_2^2 &= \left|\left| \delta \mathbf{J}^\top ( \sum_i \frac{\mu_i^\top \mu_i}{\sigma_i+1} ) \right|\right|_2^2 \\
    &= \sum_i \left( \frac{\delta \mathbf{J}^\top \mu_i^\top}{\sqrt{\sigma_i+1}} \right)^2 = \sum_i \left( \frac{||\delta \mathbf{J}^\top||_2 \cos \theta_i }{\sqrt{\sigma_i+1 }} \right)^2.
\end{split}
\end{equation}

\textbf{Final expression.}
The regularized solution in Case 2 guarantees a larger upper bound on $||\Delta||_2^2$ compared to the solution in Case 1, as the denominator $\sqrt{\sigma_i + 1}$ becomes smaller than $\sqrt{\sigma_i}$. To unify both cases, we adopt the regularized form as a conservative estimate, leading to the final expression:
\begin{equation}
\begin{split}
    ||\Delta||_2^2 \leq  \sum_i \left( \frac{||\delta \mathbf{J}^\top||_2 \cos \theta_i }{\sqrt{\sigma_i+1 }} \right)^2.
\end{split}
\end{equation}
\end{proof}
\begin{figure}[!t]
    \centering
    \subfloat[Expected PAF]{
        \includegraphics[width=0.8\linewidth]{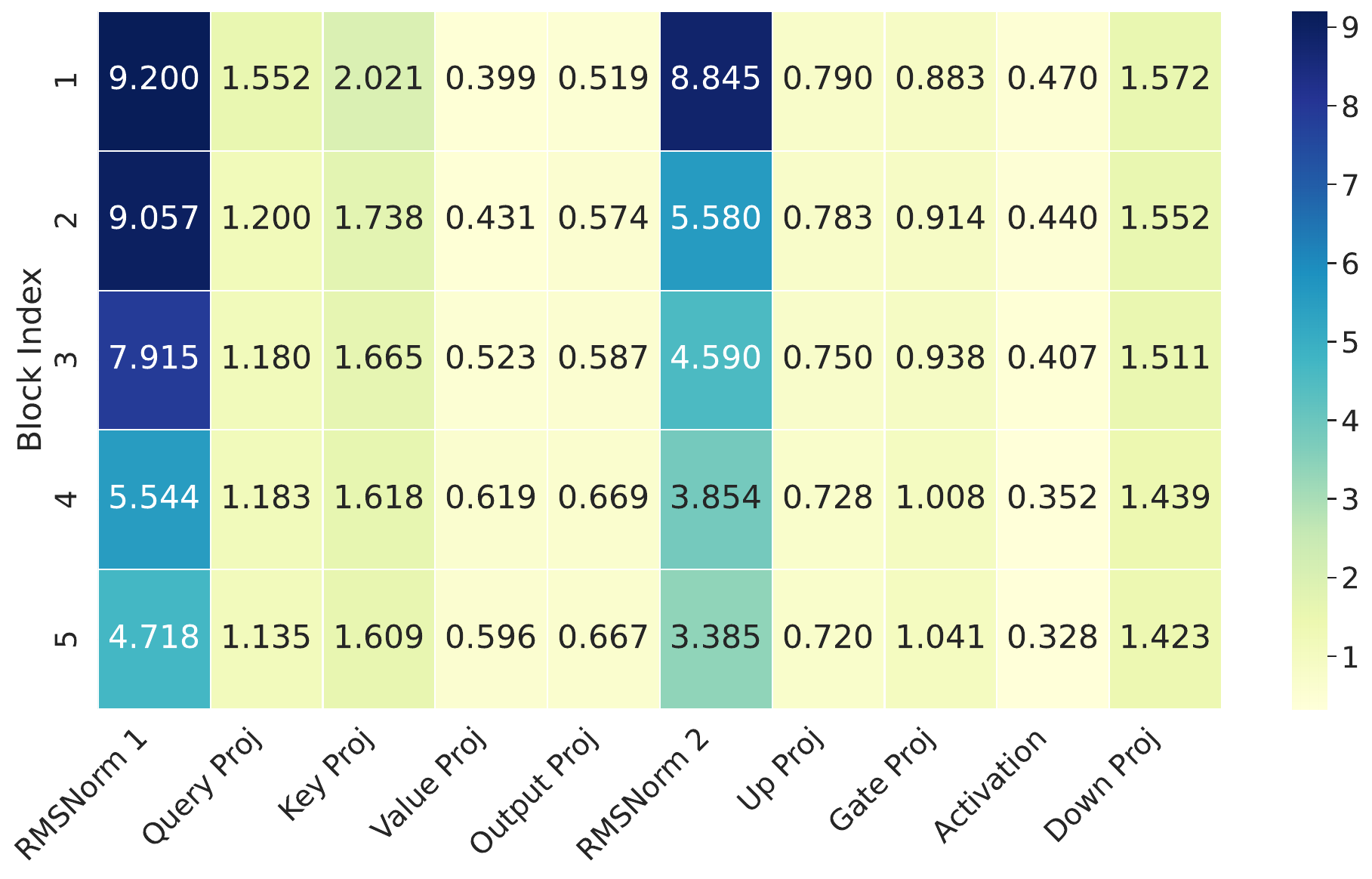}
    }
    \vfill
    \subfloat[Max-PAF]{
        \includegraphics[width=0.8\linewidth]{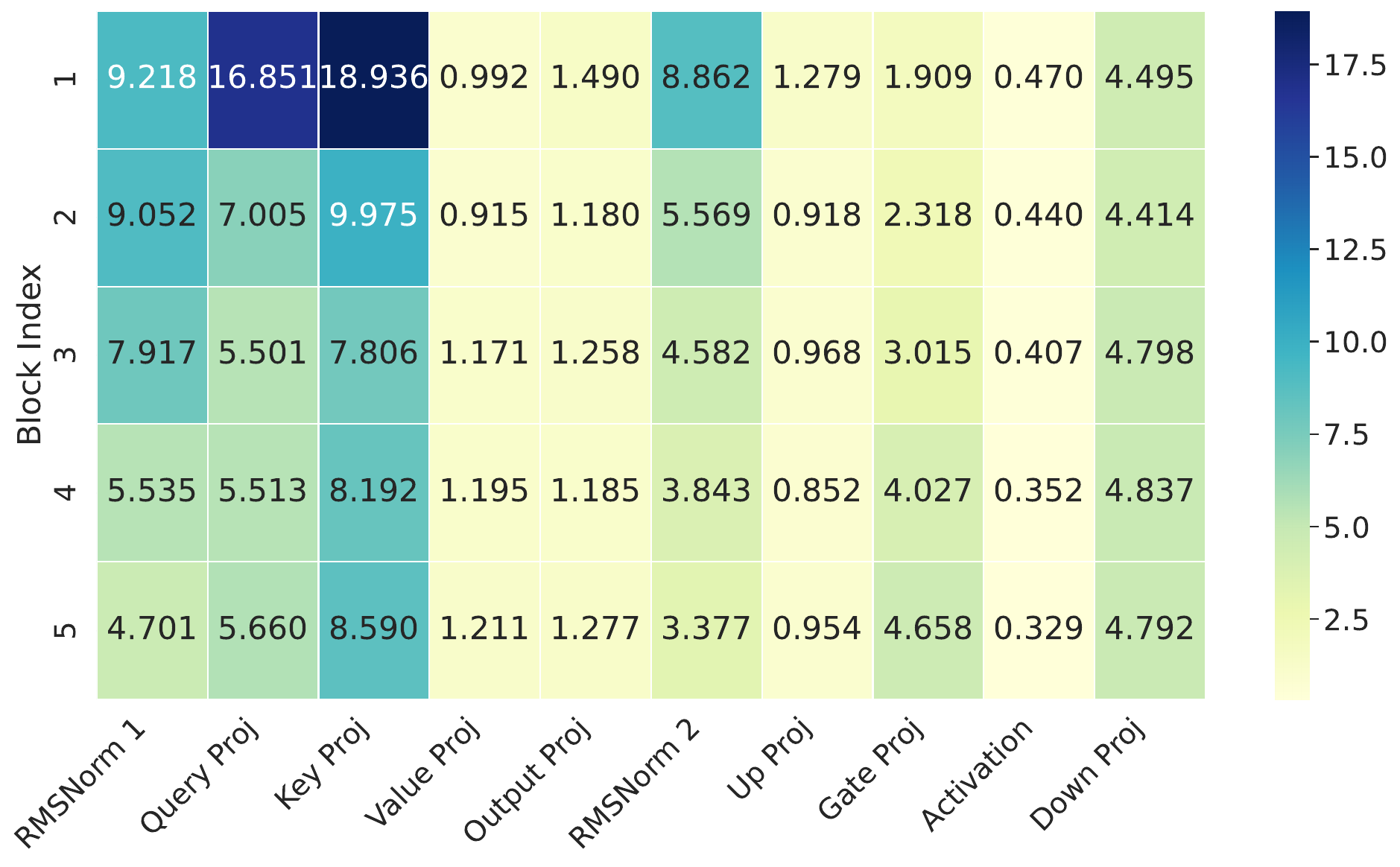}
    }
    \vspace{-4mm}
    \caption{Comparison of different layers' sensitivity in Llama-3.2-1B.}
    \label{fig_jac_lamma}
\end{figure}

\section{PAF Evaluation in Llama-3.2-1B and  SmolLM2-1.7B}
\label{appendix_paf_llama_smo}

Figures \ref{fig_jac_lamma} and \ref{fig_jac_smollm} present PAF values for different layers of Llama-3.2-1B and SmolLM2-1.7B, respectively.

\begin{figure}[!t]
    \centering
    \subfloat[Expected PAF]{
        \includegraphics[width=0.8\linewidth]{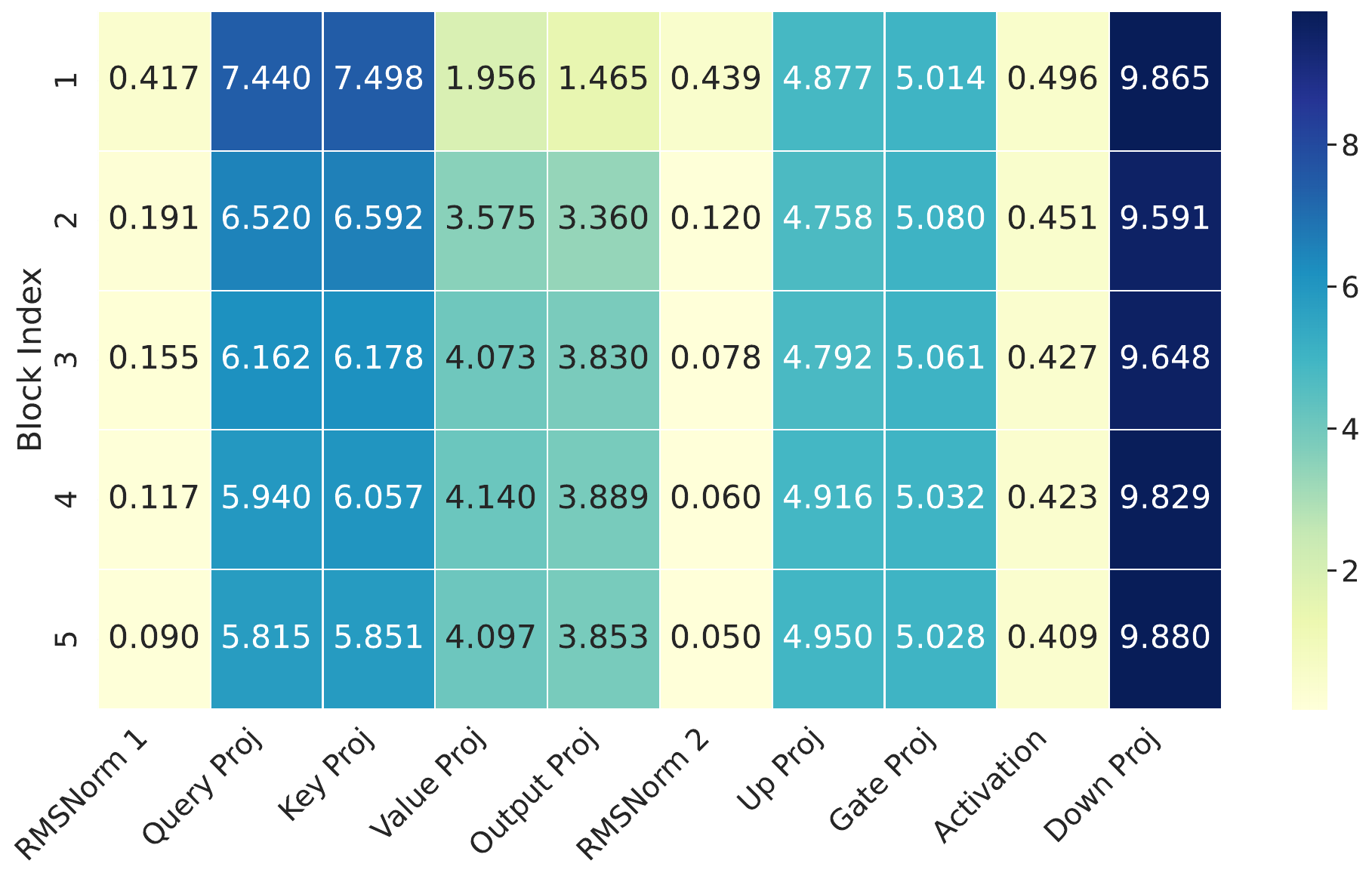}
    }
    \vfill
    \subfloat[Max-PAF]{
        \includegraphics[width=0.8\linewidth]{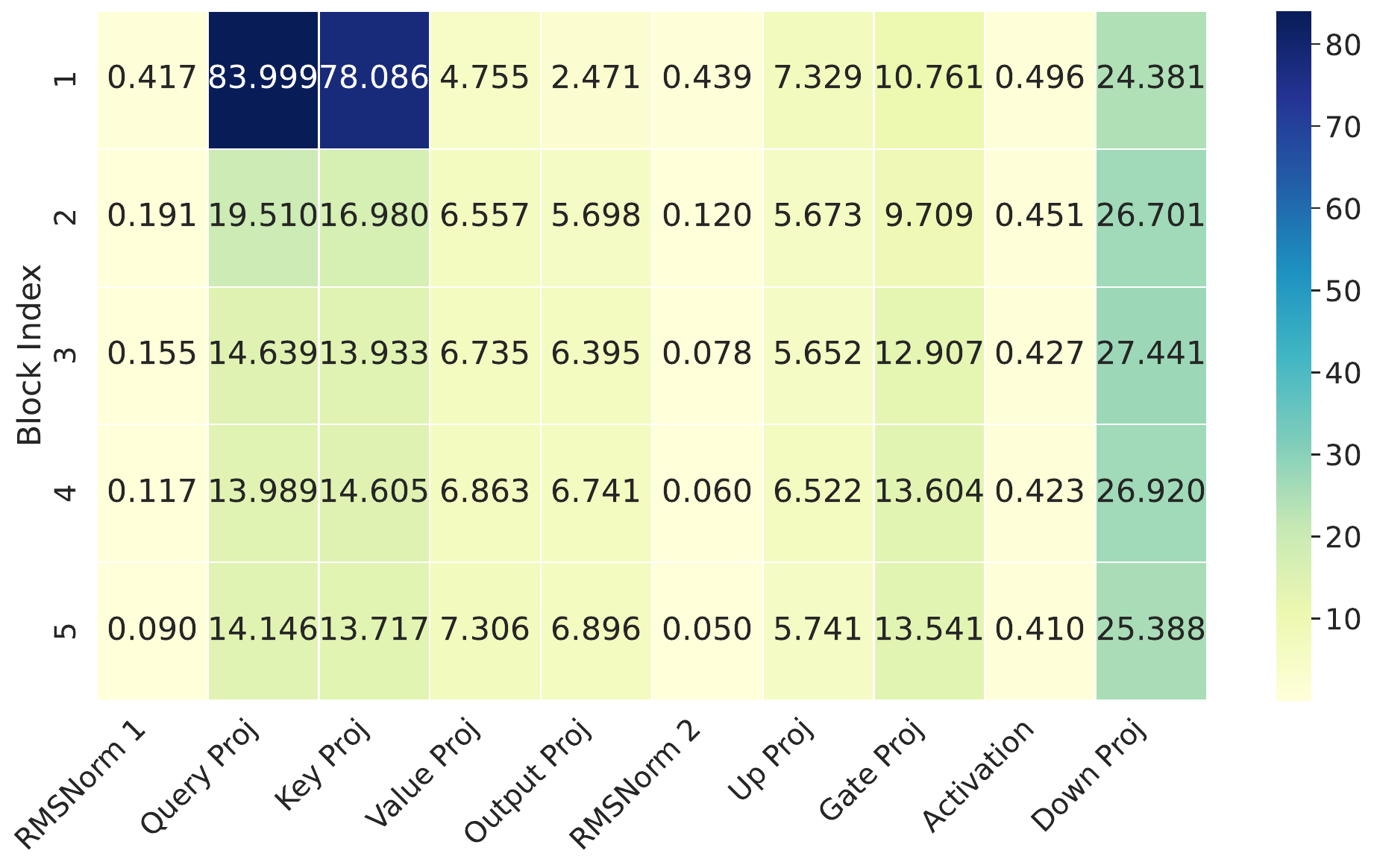}
    }
    \vspace{-4mm}
    \caption{Comparison of different layers' sensitivity in SmolLM2-1.7B.}
    \label{fig_jac_smollm}
\end{figure}

\section{Proof of Theorem \ref{thm:bound}}
\label{appendix_theorem_2_proof}

\begin{proof}
We know
$
\delta = F_C(\mathbf{\hat{z}}) - F_C(\mathbf{z}) = \int_0^1 \nabla F_C(\mathbf{z} + t\Delta) \,\Delta \; dt.
$
Taking the \(q\)-norm and using the consistency of induced norms,
\[
\|\delta\|_q \le \int_0^1 \|\nabla F_C(\mathbf{z} + t\Delta)\|_q \; dt \cdot \|\Delta\|_q \le C \|\Delta\|_q.
\]
Since \(\|\delta\|_q \le \mu\), we obtain \(\mu \le C \|\Delta\|_q\), i.e., \(\|\Delta\|_q \ge \mu/C\). This proves the first part.

For the second part, the first part gives \(\|\Delta\|_q \ge \mu/C > d_{\min}/2\). Let \(\mathbf{e}_{\min}\) be the token achieving \(d_{\min}\). By the triangle inequality,
\[
\|\mathbf{\hat{z}} - \mathbf{e}_{\min}\|_q \ge \bigl| \|\mathbf{z} - \mathbf{e}_{\min}\|_q - \|\mathbf{\hat{z}} - \mathbf{z}\|_q \bigr| = d_{\min} - \|\Delta\|_q,
\]
provided \(\|\Delta\|_q \le d_{\min}\). If \(\|\Delta\|_q > d_{\min}\), then trivially \(\mathbf{\hat{z}}\) is farther from \(\mathbf{z}\) than \(\mathbf{e}_{\min}\), so the conclusion holds. Otherwise (\(\|\Delta\|_q \le d_{\min}\)), we have \(d_{\min} - \|\Delta\|_q < \|\Delta\|_q\) because \(\|\Delta\|_q > d_{\min}/2\). Thus
\[
\|\mathbf{\hat{z}} - \mathbf{e}_{\min}\|_q \le d_{\min} - \|\Delta\|_q < \|\Delta\|_q = \|\mathbf{\hat{z}} - \mathbf{z}\|_q.
\]
Hence \(\mathbf{\hat{z}}\) is closer to \(\mathbf{e}_{\min}\) than to \(\mathbf{z}\), so the true token is not the nearest neighbor.
\end{proof}

\section{Overhead Optimization}
\label{appendix_overhead}

\begin{table}[h]
\centering
\small
\caption{Processing overhead for domain-specific terms and privacy-sensitive terms on iCliniq (per 1k tokens).}
\begin{tabular}{c|c|c}
\hline \hline
                     & \textbf{Domain-Specific} & \textbf{Privacy-Sensitive} \\ \hline
\textbf{Ratio ($r$)} & \result{19.22\%}{6.69\%} & \result{4.22\%}{4.13\%} \\
\textbf{Expected Time} & 1.58T  &  1.13T \\
\textbf{Actual Time} & \result{1.64s}{0.28s} &  \result{1.16s}{0.21s} \\
\hline \hline
\end{tabular}
\label{tab:overhead_optimization}
\end{table}

To mitigate the computational overhead of \defense, we can apply the gradient-based optimization exclusively to a subset of protected tokens, such as privacy-sensitive entities or domain-specific terms.
Let $r \in [0, 1]$ denote the protected ratio, representing the proportion of tokens subject to \defense optimization.
Let $T$ represent the time required for a standard forward pass.
Since the overhead of baseline defenses (e.g., Gaussian noise or sparsification) is negligible, their total processing time remains approximately $T$.
The expected execution time for \defense is modeled as $T + 3rT$, where the term $3rT$ accounts for the additional gradient computation and backpropagation costs associated with the optimized tokens.

Table \ref{tab:overhead_optimization} reports the processing overhead analysis on iCliniq.
We employ GPT-5 to identify domain-specific and privacy-sensitive tokens within user queries and calculate their average occurrence ratios ($r$).
We compare the theoretical expected time against the actual measured execution time, utilizing the same experimental setup described in Table \ref{tab_defense_overhead}.
The results demonstrate that, since domain-specific and privacy-sensitive terms constitute only a small fraction of the input, the actual computational overhead of \defense remains low in practice.

\end{document}